%% file: revisitingFastFourierMultiplicationAlgorithmsOnQuotientRings.tex
\RequirePackage{amsmath}
\documentclass[dvipsnames]{article} 
\usepackage[utf8]{inputenc} 
\usepackage[hyphens]{xurl} 
\usepackage{amsfonts} 
\usepackage{mathtools} 
\usepackage{bm} 
\usepackage{array} 
\usepackage[ruled,vlined]{algorithm2e} 
\usepackage[shortcuts]{extdash} 
\usepackage{hyperref} 
\usepackage{amsthm} 
\usepackage{verbatim} 
\usepackage{cleveref} 
\usepackage{enumerate}
\usepackage{orcidlink}
\usepackage{tikz} 
\usetikzlibrary{calc,arrows,positioning}

\SetKwInput{KwInput}{Input}
\SetKwInput{KwAuxiliary}{Auxiliary}
\SetKw{KwNot}{not}

\theoremstyle{definition}
\newtheorem{condition}{Condition} 
\crefname{condition}{condition}{conditions}
\Crefname{condition}{Condition}{Conditions}
\newtheorem{definition}{Definition}

\theoremstyle{plain}
\newtheorem{theorem}{Theorem}
\newtheorem{proposition}[theorem]{Proposition}
\newtheorem{corollary}[theorem]{Corollary} 
\crefname{corollary}{corollary}{corollaries}
\Crefname{corollary}{Corollary}{Corollaries}

\newtheorem{lemma}[theorem]{Lemma}
\crefname{lemma}{lemma}{lemmas}
\Crefname{lemma}{Lemma}{Lemmas}

\theoremstyle{remark}
\newtheorem{remark}{Remark}

\theoremstyle{remark}
\newtheorem*{example*}{Example}

\renewcommand{\vec}[1]{\bm{#1}}

\newcommand{\dleftarrow}{\xleftarrow{R}}

\newcommand{\Z}{\mathbb{Z}}

\newcommand{\w}{\color{white}}
\newcommand{\bigO}[1]{\mathcal{O}\left(#1\right)}
\newcommand\phantomrel[1]{\mathrel{\phantom{#1}}}

\makeatletter
\newcommand{\prem}[1]{\pod{{\operator@font rem}\mkern6mu#1}}
\makeatother
\makeatletter
\newcommand{\rem}[1]{\allowbreak\if@display\mkern18mu
  \else\mkern12mu\fi{\operator@font rem} \, \, #1}
\makeatother
\makeatletter
\newcommand{\brem}{\nonscript\mskip-\medmuskip\mkern5mu\mathbin
  {\operator@font rem}\penalty900
  \mkern5mu\nonscript\mskip-\medmuskip}
\makeatother
\makeatletter
\newcommand{\trem}{\operator@font rem}
\makeatother

\makeatletter
\newcommand{\tmod}{\operator@font mod}
\makeatother

\usepackage{hyperref}

\begin{document}

\title{Revisiting Fast Fourier multiplication algorithms on quotient rings}
\author{Ramiro Martínez\;\orcidlink{0000-0003-0496-6462}\thanks{e-mail: \nolinkurl{ramiro.martinez@upc.edu}} \qquad Paz Morillo\,\orcidlink{0000-0002-0063-2716} \\ Universitat Politècnica de Catalunya}

\date{2023\\ April}

\maketitle

\begin{abstract}
This work formalizes efficient Fast Fourier-based multiplication algorithms for polynomials in quotient rings such as $\mathbb{Z}_{m}[x]/\left<x^{n}-a\right>$, with $n$ a power of $2$ and $m$ a non necessarily prime integer. We also present a meticulous study on the necessary and/or sufficient conditions required for the applicability of these multiplication algorithms. This paper allows us to unify the different approaches to the problem of efficiently computing the product of two polynomials in these quotient rings.
\end{abstract}


%

  \input{sections/introduction}
  \input{sections/pointWiseProduct}
  \input{sections/efficientTransformation}
  \input{sections/characterizationOfPoints}
  \input{sections/generalizations}
  \input{sections/resultsAndDiscussion}

\noindent {\bf Acknowledgments.} This work is partially supported by the European Union PROMETHEUS project (Horizon 2020 Research and Innovation Program, grant 780701) and the Spanish Ministry of Economy and Competitiveness, through Project MTM2016-77213-R.


\input{revisitingFastFourierMultiplicationAlgorithmsOnQuotientRings.bbl}
\end{document}

%% file: sections/introduction.tex

\section{Introduction}

Constructing efficient multiplication algorithms for polynomials with coefficients in a ring $R$ has been an extensive research area. Given two polynomials $g(x),\allowbreak h(x) \in R[x]$ of degree bounded by $n$, $g(x) = \sum_{i=0}^{n-1} g_{i}x^{i}$ and $h(x) = \sum_{i=0}^{n-1} h_{i}x^{i}$, computing its product in a naïve way (known as the schoolbook multiplication algorithm),
\[(g \cdot h)(x) = \sum_{i=0}^{n-1}\sum_{j=0}^{n-1} \left(g_{i} \cdot h_{j}\right)x^{i+j},\]
requires a quadratic number, $n^{2}$, of multiplications of elements from the ring $R$.

\renewcommand{\arraystretch}{1.2}
  \[\arraycolsep=1pt
  \begin{array}{r<{x^6} c r<{x^5} c r<{x^4} c r<{x^3} c r<{x^2} c r<{x} c r}
              \w&   &             \w&   &             \w&   &           g_3 & + & g_2           & + & g_1           & + & g_0           \\
              \w&   &             \w&   &        \w& \times &           h_3 & + & h_2           & + & h_1           & + & h_0           \\ \hline
              \w&   &             \w&   &             \w&   & g_3 \cdot h_0 & + & g_2 \cdot h_0 & + & g_1 \cdot h_0 & + & g_0 \cdot h_0 \\
              \w&   &             \w&   & g_3 \cdot h_1 & + & g_2 \cdot h_1 & + & g_1 \cdot h_1 & + & g_0 \cdot h_1 &   &               \\
              \w&   & g_3 \cdot h_2 & + & g_2 \cdot h_2 & + & g_1 \cdot h_2 & + & g_0 \cdot h_2 &   &             \w&   &               \\
  g_3 \cdot h_3 & + & g_2 \cdot h_3 & + & g_1 \cdot h_3 & + & g_0 \cdot h_3 &   &             \w&   &             \w&   &               \\ \hline
  \textstyle\sum\limits_{\scalebox{0.4}{i+j=6}}g_{i}h_{j} & + &
  \textstyle\sum\limits_{\scalebox{0.4}{i+j=5}}g_{i}h_{j} & + &
  \textstyle\sum\limits_{\scalebox{0.4}{i+j=4}}g_{i}h_{j} & + &
  \textstyle\sum\limits_{\scalebox{0.4}{i+j=3}}g_{i}h_{j} & + &
  \textstyle\sum\limits_{\scalebox{0.4}{i+j=2}}g_{i}h_{j} & + &
  \textstyle\sum\limits_{\scalebox{0.4}{i+j=1}}g_{i}h_{j} & + &
  g_{0}h_{0}
  \end{array}\]

If we now want to work in a ring $R[x]/\left<f(x)\right>$ the schoolbook multiplication algorithm would perform the same amount of operations taking into account that at the end we should perform a reduction.

  For the particular cases $f(x) = x^n \pm 1$ it is clear that the algorithm performs the same number of operations, as we just need to multiply by $\mp 1$ all $i$-coefficients from $n$ to $2 n - 2$ and add them to their corresponding $i-n$ column.

The goal of this work is to formalize and unify some concepts used to build more efficient multiplication algorithms focusing on the ring $R[x] = \Z_{m}[x]$ and then on its quotient $\Z_{m}[x]/\left<x^n+1\right>$, with $n$ a power of $2$ and $m$ non necessarily prime (although sometimes we would also consider $x^n-1$ or in general $x^{n}-a$).

We choose to study optimizations for this particular ring as it is widely used by cryptographic constructions that base their security on ideal lattices~\cite{EC:LyuPeiReg10}, that can be identified with ideals in $\Z_{m}[x]/\left<x^n+1\right>$.

Most of the literature usually deals with this matter by providing a set of recipes that can be applied for some specific particular rings (considering whether $m$ is prime or not or how does $x^{n}-a$ split), without specifying if the imposed conditions in these recipes are necessary or only sufficient. This lack of detail might make more difficult the applicability of such recipes. Through this article we instead analyze what are the fundamental properties that allow us to obtain a computational speedup from a comprehensive and mathematical point of view, so that the reader can apprehend these techniques and distinguish intrinsic properties from superfluous conventions.

Therefore, the analysis presented in this paper will help the reader to avoid confusions when using multiplication algorithms for polynomials in quotient rings $\Z_{m}[x]/\left<x^{n}+1\right>$.

 Many of the ideas developed in this article are folklore when working on other rings such as $\mathbb{C}[x]$ or $\Z_{q}[x]$ with $q$ a prime satisfying certain conditions, but an exhaustive analysis might be very helpful to analyze when these ideas can be, completely or partially, generalized to our ring of interest and why the required conditions for the underlying ring are indeed necessary or just sufficient.

\subsection{Related works}

\subsubsection{Karatsuba multiplication algorithm}

The first subquadratic multiplication algorithm was designed by Karatsuba in~\cite{karatsuba1963multiplication}, with a cost of $\bigO{n^{\log{3}}}$ derived from a clever divide an conquer strategy. We include it here as our approach uses Karatsuba's algorithm as a subroutine.

Let $g(x)$ and $h(x)$ be two polynomials of degree strictly smaller than $n$ (a power of two). We can split these polynomials into upper and lower degree polynomials as $g(x) = g_{U}(x)x^{n/2} + g_{L}(x)$ and $h(x) = h_{U}(x)x^{n/2} + h_{L}(x)$ where $g_{L}, g_{U}, h_{L}, h_{U}$ have all degree smaller than $n/2$.

  A naïve computation would be
  \begin{align*}
    g(x) \cdot h(x) &= (g_{U}(x)x^{n/2} + g_{L}(x))(h_{U}(x)x^{n/2} + h_{L}(x)) \\
     &= \left(g_{U}(x) \cdot h_{U}(x)\right)x^n \\
     &\phantom{{}={}}+ \left(g_{U}(x) \cdot h_{L}(x) + g_{L}(x) \cdot h_{U}(x)\right)x^{n/2} \\
     &\phantom{{}={}}+ g_{L}(x) \cdot h_{L}(x).
  \end{align*}

  That way we divide the full multiplication into four multiplications of polynomials of size $n/2$. One can see however that this does not improve the efficiency of the computation. Let $T(n)$ be the number of operations required for computing the product using this method. The recurrence obtained, $T(n) = 4 T(n/2) + \bigO{n}$, implies (via the Master Theorem for divide-and-conquer recurrences~\cite{masterTheorem}) that $T(n) = \bigO{n^2}$.

Karatsuba's gifted idea was to notice that the crossed terms can be obtained from the other terms and a single multiplication of $n/2$-polynomials. That is, we can write the term $\left(g_{U}(x) \cdot h_{L}(x) + g_{L}(x) \cdot h_{U}(x)\right)$ as
\[(g_{L}(x) + g_{U}(x))(h_{L}(x) + h_{U}(x)) - \left(g_{U}(x) \cdot h_{U}(x)\right) - \left(g_{L}(x) \cdot h_{L}(x)\right) .\]

\begin{algorithm}[H]
  \linespread{1.35}\selectfont
   \caption{\textsc{Karatsuba}}\label{alg:karatsuba}
\SetAlgoLined{}
\KwInput{Two polynomials $g(x)$ and $h(x)$ of degree bounded by $n$}
\KwResult{Product of $g(x) \cdot h(x)$}
\nl{} \lIf{n = 1}{\Return{$g \cdot h$}}
Split $g(x)$ and $h(x)$ into $g_{L}(x), g_{U}(x), h_{L}(x), h_{U}(x)$.\\
\nl{} $a(x) \coloneqq \textsc{Karatsuba}(g_{U}(x),h_{U}(x))$\\
\nl{} $b(x) \coloneqq \textsc{Karatsuba}(g_{L}(x),h_{L}(x))$\\
\nl{} $c(x) \coloneqq \textsc{Karatsuba}(g_{L}(x)+g_{U}(x),h_{L}(x)+h_{U}(x))$\\
\nl{} \Return{$a(x)x^{n} + (c(x)-a(x)-b(x))x^{n/2} + b(x)$}
\end{algorithm}

Notice how the recurrence now computes only $3$ products of half size and a linear amount of operations ($T(n) = 3T(n/2) + \bigO{n}$) providing the desired sublinear running time of $T(n) = \bigO{n^{\log{3}}}$ (solving the recurrence with the Master Theorem for divide-and-conquer recurrences~\cite{masterTheorem}).

  \subsubsection{Karatsuba multiplication algorithm \texorpdfstring{$\mod f(x)$}{mod f(x)}}

    However notice this is not the most natural way of writing this recursion when working $\mod x^n \pm 1$ as all recursive calls work the same way but the last one, where a reduction has to be performed.

    Alternatively we can split $g(x) = g_{1}(x^2)x + g_{0}(x^2)$, with $g_{0}$ containing the even coefficients and $g_{1}$ containing the odd ones. This allows us to think of $g(x) \in R[x]$ of degree smaller than $n$ as $g(x,y) = g_{0}(y) + g_{1}(y)x \in R[x,y]$ of $x$-degree $1$ and $y$-degree smaller than $n/2$. It is called \emph{Dual Karatsuba} and the idea remains the same:
    \begin{align*}
      g(x) \cdot h(x) &= (g_{1}(x^2)x + g_{0}(x^2))(h_{1}(x^2)x + h_{0}(x^2)) \\
       &= \left(g_{1}(x^2) \cdot h_{1}(x^2)\right)x^2 \\
       &\phantom{{}={}}+ \left(g_{1}(x^2) \cdot h_{0}(x^2) + g_{0}(x^2) \cdot h_{1}(x^2)\right)x \\
       &\phantom{{}={}}+ g_{0}(x^2) \cdot h_{0}(x^2)
     \end{align*}

     And analogously as we did before we can write the second term with only one additional multiplication. $\left(g_{1}(x^2) \cdot h_{0}(x^2) + g_{0}(x^2) \cdot h_{1}(x^2)\right)$ is:
     \[\left(\left(g_{1}(x^2) + g_{0}(x^2)\right)\cdot\left(h_{1}(x^2) + h_{0}(x^2)\right) - g_{1}(x^2) \cdot h_{1}(x^2) - g_{0}(x^2) \cdot h_{0}(x^2) \right).\]

     \begin{algorithm}[H]
       \linespread{1.35}\selectfont
        \caption{\textsc{Karatsuba} $\pmod{x^n \pm 1}$}\label{alg:dualKaratsuba}
     \SetAlgoLined{}
     \KwInput{Two polynomials $g(x)$ and $h(x)$ of degree bounded by $n$}
     \KwResult{Product of $g(x) \cdot h(x)$}
     \lIf{n = 1}{\Return{$g \cdot h$}}
     Split $g(x)$ and $h(x)$ into $g_{0}(x), g_{1}(x), h_{0}(x), h_{1}(x)$.\\
     $a(y) \coloneqq \textsc{Karatsuba}(g_{1}(y),h_{1}(y))$\\
     $b(y) \coloneqq \textsc{Karatsuba}(g_{0}(y),h_{0}(y))$\\
     $c(y) \coloneqq \textsc{Karatsuba}(g_{0}(y)+g_{1}(y),h_{0}(y)+h_{1}(y))$\\
     \Return{$a(x^2)x^2 + (c(x^2)-a(x^2)-b(x^2))x + b(x^2) \pmod{x^n \pm 1}$}
     \end{algorithm}

     Now the reduction modulo $x^n \pm 1$ works at each level of the recursive algorithm, as $\widehat{g}(y) \cdot \widehat{h}(y) \pmod{y^{n/2} \pm 1}$ is equivalent to $\widehat{g}(x^2) \cdot \widehat{h}(x^2) \pmod{x^{n} \pm 1}$ once we change variables again. We obtain no computational advantage, but it allows us to understand it from a different perspective.

     Notice $g_{0}(y)+g_{1}(y) \equiv g(x,y) \mod x-1$ and $g_{0}(y) \equiv g(x,y) \mod x$. These ideas are considered in~\cite{Bernstein01} in order to see all these tools as part of the same framework.

\subsubsection{Faster multiplication algorithms}

Even faster algorithms can be obtained from more clever recurrences, mapping the polynomials into a different domain in a recursive way (recursively computing two transforms of half the size $T(n) = 2 T(n/2) + \bigO{n}$) where they can be efficiently multiplied in linear time, so that the final computational complexity is $\bigO{n \log n}$. Through this work we are going to focus and systematically explore the ones derived from the Fast Fourier Transform (FFT) paradigm, that is going to be extensively described in the following sections. This approach is usually referred as Number Theoretic Transformation (NTT) when working with finite fields.

The main idea of the FFT recurrence is attributed to Gauss and was fully developed by Cooley and Tukey in its seminal work~\cite{cooley1965algorithm} considering the ring of complex numbers.

Many variants have been developed since then, generalizing~\cite{cooley1965algorithm} to non-power of two bounded degree polynomials or providing additional tricks and interpretations from which we benefit, such as~\cite{STOC:fiduccia72}. An extensive and magnificently well documented survey can be found in~\cite{Bernstein01}.

Most of the work has focused on the particular case of multiplications in $\Z_{q}[x]/\left<x^{n}+1\right>$ when $q$ is prime and $x^n+1$ fully splits in linear factors. That has been studied for a while, both from a software~\cite{FSE:LMPR08} and hardware~\cite{LC:PopGun12} point of view.

We are particularly interested in the ideas presented in~\cite{EC:LyuSei18}, as they specifically discuss multiplications of polynomials in $\Z_{q}[x]/\left<x^{n}+1\right>$ when $x^n+1$ does not fully split in linear factors (a situation that happens in some lattice-based cryptographic schemes such as some commitment schemes~\cite{SCN:BDLOP18}) and the standard FFT can only be partially applied. However~\cite{EC:LyuSei18} only considers the case with a prime $q$ and briefly describes the procedure.

This technique of partially applying an FFT is sometimes called incomplete NTT~\cite{TCHES:CHKSS21} and usually interpreted like a Chinese Remainder Transform (CRT), as in~\cite{EC:LyuPeiReg13} doing Fast Chinese remaindering~\cite{gathen_gerhard_2013}.

However, most of the literature only provides some sufficient conditions that allow some particular implementation or specific abstraction of a fast multiplication algorithm that are not directly generalizable. In spite of that we present a more general framework for multiplications in $\Z_{m}[x]/\left<x^{n}+1\right>$ that would allow the reader to comprehend why some folklore assumptions are indeed necessary and why some others are not, from a mathematically rigorous, yet accessible for readers not familiarized with algebraic constructions, point of view.

\subsection{Notation and conventions}

Since our goal is to work in $R[x]/\left<f(x)\right>$, with $f$ a monic polynomial, we choose as a representative for $g(x) \in R[x]/\left<f(x)\right>$ its remainder when divided by $f(x)$, denoted by $g(x) \brem f(x)$.

We denote vectors by lower-case bold-faced roman letters and use $\log$ for the binary logarithm.

We borrow most of our notation from~\cite{Bernstein01}, and present some new definitions through \cref{sec:efficientTrans,sec:suitableSets}, that we believe are of independent interest.

%% file: sections/pointWiseProduct.tex

\section{Pointwise product}

The main idea behind any Fast Fourier Transform multiplication technique is to compute the product of two polynomials via the pointwise product of their evaluations on certain points of the ring $R$.

It is straightforward by the definition of the product of polynomials that given a point $x_{0} \in R$ the evaluation of the product is equal to the product of the evaluations $(g \cdot h)(x_{0}) = g(x_{0}) \cdot h(x_{0})$.

This way if we compute enough evaluations (we denote this transform as $T$) then we could perform a pointwise product (denoted by $\odot$) of the evaluations and then interpolate back ($T^{-1}$) the polynomial.

\[
  \begin{array}{*{7}{c}}
 {(R[x])}^2 &\xrightarrow{\odot\circ T} &R^k &\xrightarrow{T^{-1}} &R[x] \\
   (g(x),h(x))
 &\mapsto
 &\begin{pmatrix} g(x_{0})\\ g(x_{1}) \\ \vdots \\ g(x_{k-1}) \end{pmatrix}\odot
 \begin{pmatrix} h(x_{0})\\ h(x_{1}) \\ \vdots \\ h(x_{k-1}) \end{pmatrix}
 &\mapsto
   &(g \cdot h)(x)
 \end{array}
\]

In the following sections we will explore what possibilities do we have for $R$, $k$ and $x_{0}, x_{1}, \dots,\allowbreak x_{k-1}$ so that $T^{-1}$ is well defined and both $T$ and $T^{-1}$ are efficiently computable. We are going to characterize the necessary and sufficient conditions the evaluation points have to satisfy to be able to perform these operations.

\subsection{General invertibility of \texorpdfstring{$T$}{T}}

To deal with the invertibility of transform $T$ when applied to polynomials of degree bounded by $n$ we notice it is a linear mapping and characterize it by its associated matrix

\[V =
  \begin{pmatrix}
    1      & x_{0}   & {x_{0}}^2   & \cdots  & {x_{0}}^{n-1} \\
    1      & x_{1}   & {x_{1}}^2   & \cdots  & {x_{1}}^{n-1} \\
    \vdots & \vdots  & \vdots    & \ddots & \vdots \\
    1      & x_{k-1} & {x_{k-1}}^2 & \cdots  & {x_{k-1}}^{n-1} \\
  \end{pmatrix}.
\]

This special matrix is known as a Vandermonde matrix. If we choose $k = n$ we have a square Vandermonde matrix and we can discuss its invertibility.

Since we are working with a general commutative ring with unity $R$, a matrix is invertible if and only if its determinant is invertible~\cite{alma991004101649706711}. The determinant of a Vandermonde matrix is easy to compute and has the form
\[ \det(V) = \prod_{0 \leq i < j < n}(x_{j} - x_{i}).\]
Now if $R$ is a field it is just sufficient to choose $n$ different evaluation points $x_{0}, \dots, x_{n-1}$.

If $R$ is just a commutative ring it is a necessary and sufficient condition to choose points such that their differences are invertible in $R$.

\begin{condition}[Points with invertible differences]\label{cond:invertibleDifferences}
  We say a set of points satisfies \Cref{cond:invertibleDifferences} if the difference of every pair is invertible in $R$.
\end{condition}

\begin{remark}
  Choosing $n$ evaluation points with invertible differences in $R$ is a necessary and sufficient condition for the transform $T$ to be invertible.
\end{remark}

Notice how we are talking about transforming and anti-transforming polynomials of a certain degree. Given two polynomials $g(x)$ and $h(x)$ of degrees $n$ and $n'$, since we want to recover the polynomial $g(x) \cdot h(x)$, we would have to think them as polynomials of degree smaller or equal than $n+n'$ and use $(n+n'+1) \times (n+n'+1)$ Vandermonde matrices.

\subsection{Pointwise product of evaluations modulo \texorpdfstring{$f(x)$}{f(x)}}

Given $a\in R$\footnote{
One can see that in the special cases $a = \pm 1$ we have that the product in  $R[x]/\left<x^n-1\right>$ and $R[x]/\left<x^n+1\right>$ is usually described in the literature as a cyclic/anti-cyclic convolution (denoted by $\ast$).

\begin{align*}
  (g \cdot h)(x) \prem{x^n - 1} &= g(x) \cdot h(x) \prem{x^n - 1} \\
  &= \left(\sum_{i=0}^{n-1} g_{i}x^{i}\right)\left(\sum_{j=0}^{n-1} h_{j}x^{j}\right) \prem{x^n - 1}\\
  &= \sum_{k=0}^{n-1}\left(\sum_{\substack{i,j\\i+j \equiv k \\ \mod{n}}} g_{i} \cdot h_{j}\right)x^{k}\\
  &= (g \ast h)(x)
\end{align*}

This intuition might be of independent interest as a convolution product in the regular domain is a pointwise product in the transformed domain.
}
 we can always consider polynomials in $R[x]/\left< x^n - a \right>$ as polynomials in $R[x]$ with degree strictly bounded by $n$ (using the canonical representative $\trem$), compute, as we said before, their product as a polynomial in $R[x]$ with degree strictly bounded by $2n$ via a $2n$-transform and then applying $\rem x^{n}-a$ again to obtain the representative with degree smaller than $n$.

The main issue we face when trying to use the pointwise product of evaluations technique to compute the product of two polynomials modulo $f(x)$ is that evaluation is not well defined in general as it depends on the representative we choose from the class of equivalence modulo $f(x)$.

In general, for an arbitrary $x_{0}$, it is not the same to compute $(g(x) \brem f(x))(x_{0}) \cdot (h(x) \brem f(x))(x_{0})$ than $(g(x) \cdot h(x) \brem f(x))(x_{0})$.

For this reason we have to choose specific points where evaluation is compatible with congruence classes.
If $\alpha$ is such that for any two equivalent polynomials $g(x) \equiv \widehat{g}(x) \pmod{f(x)}$ we obtain $g(\alpha) = \widehat{g}(\alpha)$ then we necessarily have $f(\alpha) = 0$ and $\alpha$ has to be a root of $f(x)$.

Let $\alpha$ be a root of $f(x)$. Then $ g(x) \equiv \widehat{g}(x) \pmod{f(x)}$ implies $g(\alpha) = \widehat{g}(\alpha)$ and therefore
\[(g(x) \brem f(x) )(\alpha) \cdot (h(x) \brem f(x))(\alpha) = (g(x) \cdot h(x) \brem f(x))(\alpha).\]

In conclusion, for the particular case with $f(x) = x^{n} - a$, where $a \in R$, $\alpha$ has to be an $n$\=/th root of $a$. By choosing $\alpha_{0},\dots,\alpha_{n-1}$ different $n$\=/th roots of $a$ with invertible differences we would be able to directly recover $g(x) \cdot h(x) \brem {x^{n} - a}$ from the pointwise product of their evaluations in $\alpha_{0},\dots,\alpha_{n-1}$.

This reduces the number of computations needed, since only $n$ evaluation points are required and no further reduction is needed to obtain the canonical representative.

\begin{condition}[Roots of $f(x)$ as points]\label{cond:evalAtRoots}
  We say a set of points in $R$ for a polynomial in $R[x]/\left<f(x)\right>$ satisfies \Cref{cond:evalAtRoots} if they are roots of $f(x)$.
\end{condition}

\begin{remark}
  Choosing $n$\=/th roots of $a$ as evaluation points is a necessary and sufficient condition for the evaluations of polynomials in $R[x]/\left<x^{n}-a\right>$ to be well defined.
\end{remark}

So far \cref{cond:invertibleDifferences,cond:evalAtRoots} are necessary and sufficient conditions to use pointwise product of evaluations as a technique to compute the product of two polynomials in $R[x]/\left<x^n-a\right>$. The next step is to study when can this be computed efficiently.

%% file: sections/efficientTransformation.tex

\section{Efficient transforms}\label{sec:efficientTrans}

In this section we are going to define efficient transform and anti-transform protocols from a theoretical and asymptotical point of view. Additional implementation tricks or approaches (for example, whether the recurrences are solved in an iterative or recursive way) could have an important impact to save up space or computations but are out of the scope of this article.

\subsection{Efficient evaluation of \texorpdfstring{$T$}{T}}

Once we have seen the requirements for the pointwise product of evaluations to work under each possible concerned circumstances we have to discus how to efficiently apply them.

Computing each of the $n$ evaluations individually would require $\bigO{n}$ operations, for a total of $\bigO{n^2}$. The Fast Fourier approach outperforms that computing the $n$ evaluations at the same time by means of a divide-and-conquer recursive strategy.

If we call $x_{i}$ to one of the evaluation points and we decompose the polynomial $g(x)$ into two polynomials $g_{0}(x)$ and $g_{1}(x)$ of half the size with even and odd coefficients respectively, as in Dual Karatsuba, we can write
\[
  g(x_{i})   = g_{0}(x_{i}^2) + x_{i} \cdot g_{1}(x_{i}^2).
\]

With this recursion we reduce a single polynomial evaluation of degree bounded by $n$ to two polynomial evaluations of degree bounded by $n/2$, a product in $R$ and an addition in $R$. Directly doing this would not save us any cost, as it would still take $\bigO{n}$ per evaluation.

The main idea is to choose the evaluation points $\{x_{0}, x_{1}, \dots, x_{n-1}\}$ so that the set containing their squares $\{y \; \vert \; y = {x_{i}}^2\}$ contains only $n/2$ elements (therefore we could reuse the evaluations of $g_{0}$ and $g_{1}$ on the squares). We would like that to be true recursively so we introduce the following definition, already satisfying \Cref{cond:evalAtRoots}.

\begin{definition}[Twofold set of $n$\=/th roots]
  An indexed set $\alpha_{0}, \dots, \alpha_{n-1} \in R$ (properly reindexed if required) of $n$\=/th roots of an element $a \in R$, with $n$ a power of 2, is said to be a \emph{twofold set of $n$\=/th roots} if $i \equiv j \pmod{2^{{\log(n)}-k}}$ implies ${\alpha_{i}}^{2^k} = {\alpha_{j}}^{2^k}$ for $k$ from $0$ to $\log(n)$.
\end{definition}

We can visually represent it as a full binary tree like in \Cref{fig:binarytree}. Observe the evaluation points, leafs in the tree, appear in \emph{bit-reversed order}.

\begin{figure}[ht]
 \begin{minipage}{\textwidth}
   \centering
    \begin{tikzpicture}[every node/.style={minimum size=0.85cm,circle,draw,line width=1pt,color=black},edge from parent/.style={draw,line width=1pt,<-},
      level/.style={sibling distance=60mm/#1}, level distance=4.25em,
      edge from parent path={
      (\tikzparentnode) |-   
      ($(\tikzparentnode)!0.5!(\tikzchildnode)$) -| 
      (\tikzchildnode)},
      scale=1,transform shape
      ]
      \footnotesize
      \node (a){$a$}
        child {node (a0) {$\alpha_{\,0}$}
          child {node (a00) {$\alpha_{\,00}$}
            child {node (a000) {$\alpha_{\,000}$} edge from parent}
            child {node (a100) {$\alpha_{\,100}$}  edge from parent}
          }
          child {node (a10) {$\alpha_{\,10}$}
            child {node (a010) {$\alpha_{\,010}$}}
            child {node (a110) {$\alpha_{\,110}$}}
          }
        }
        child {node (a1) {$\alpha_{\,1}$}
          child {node (a01) {$\alpha_{\,01}$}  edge from parent
            child {node (a001) {$\alpha_{\,001}$}  edge from parent}
            child {node (a101) {$\alpha_{\,101}$}  edge from parent}
          }
          child {node (a11) {$\alpha_{\,11}$}  edge from parent
            child {node (a011) {$\alpha_{\,011}$}  edge from parent}
            child {node (a111) {$\alpha_{\,111}$}  edge from parent}
          }
      };
    \end{tikzpicture}
    \caption{Twofold set of $8$\=/th roots}\label{fig:binarytree}
  \end{minipage}
\end{figure}
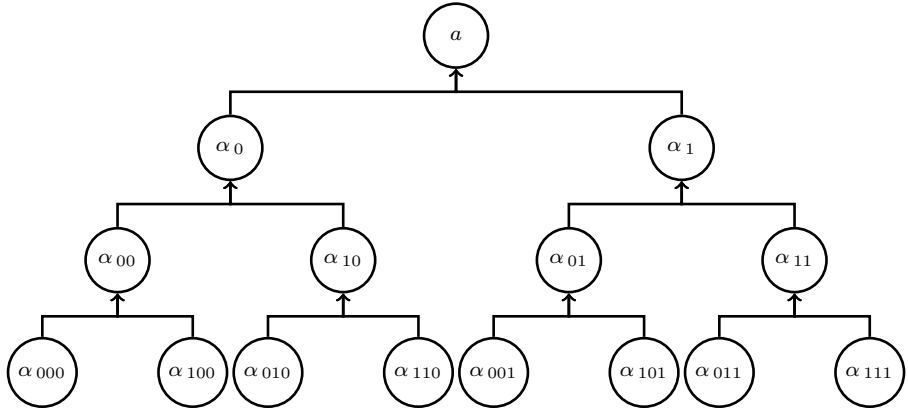

\begin{remark}\label{rem:twofoldSquares}
  Given $\alpha_{0}, \dots, \alpha_{n-1}$ a twofold set of $n$\=/th roots of $a \in R$ then the set $\{ {\alpha_{0}}^{2}, \dots,\allowbreak {\alpha_{n/2-1}}^{2} \}$ is a twofold set of $n/2$\=/th roots of $a$.
\end{remark}

Choosing as evaluation points a twofold set of roots $\alpha_{0}, \dots, \alpha_{n-1}$ we have that for every $0 \leq i < n/2$ the equality ${\alpha_{i}}^2 = {\alpha_{i + n/2}}^2$ holds and we can write

\begin{align*}
g(\alpha_{i})   &= g_{0}({\alpha_{i}}^{2})   + \phantom{{}^{{}+ n/2}}\alpha_{i} \cdot g_{1}({\alpha_{i}}^{2}) \\
g(\alpha_{i+n/2}) &= g_{0}({\alpha_{i}}^{2}) + \alpha_{i+n/2} \cdot g_{1}({\alpha_{i}}^{2}).\\
\end{align*}

We can use it to present our first general description of a Fast Fourier Transform (FFT) \Cref{alg:fftRootsOfA}.

\begin{algorithm}[H]
  \linespread{1.35}\selectfont
   \caption{\textsc{FFT}}\label{alg:fftRootsOfA}
\SetAlgoLined{}
\KwInput{A polynomial $g(x)$ of degree bounded by $n$ and a twofold set $\alpha_{0}, \dots, \alpha_{n-1}$ of $n$\=/th roots of $a$ }
\KwResult{Evaluations of $g(x)$ at $\alpha_{0}, \dots, \alpha_{n-1}$}
\nl{} \lIf{n = 1}{\Return{$g$}}
Split $g(x)$ into $g_{0}(x)$ and $g_{1}(x)$.\\
\nl{} $\vec{y}_{0} \coloneqq \textsc{FFT}(g_{0}(x),{\alpha_{0}}^{2}, \dots, {\alpha_{n/2-1}}^{2})$\\
\nl{} $\vec{y}_{1} \coloneqq \textsc{FFT}(g_{1}(x),{\alpha_{0}}^{2}, \dots, {\alpha_{n/2-1}}^{2})$\\
\nl{} \Return{$\vec{y}_{0} + \vec{y}_{1} \odot (\alpha_{0}, \dots, \alpha_{n/2-1}) \vert \vert \vec{y}_{0} + \vec{y}_{1} \odot (\alpha_{n/2}, \dots, \alpha_{n-1})$}
\end{algorithm}

Analyzing its computational cost now we find that computing the $n$-FFT takes as much time as computing two $n/2$-FFT plus a linear amount of products and additions in $R$ ($T(n) = 2T(n/2) + \bigO{n}$). Using again~\cite{masterTheorem} we end up with a total cost $\bigO{n\log{n}}$. Observe that now this is the total cost of the transform and not per evaluation as it was the case before.

\begin{condition}[Twofold set of roots]\label{cond:twofold}
  We say a set of points satisfies \Cref{cond:twofold} if it is a twofold set.
\end{condition}

\begin{remark}
  Choosing the evaluation points as a twofold set of $n$\=/th roots of $a \in R$ is a sufficient condition for the existence of an efficient FFT\@.
\end{remark}

Observe that, by definition, \Cref{cond:twofold} implies \Cref{cond:evalAtRoots}. However we prefer to treat it separately as it is sufficient for an efficient implementation but unnecessary for a general pointwise product method.

We are going to use the following notation. For any $i \in \{0,1,\dots,n-1\}$ let $\overline{\imath} = i + n/2 \brem n$. Analogously with $j$ and $\overline{\jmath}$.

\subsection{Efficient evaluation of \texorpdfstring{$T^{-1}$}{T⁻¹}}

In order to obtain an efficient multiplication algorithm we need not only an efficient transform algorithm but also an efficient anti-transform algorithm.

Notice beforehand that, the same way we saw in \Cref{rem:twofoldSquares} that \Cref{cond:twofold} is preserved when squaring the evaluation points we should check if the same holds with \Cref{cond:invertibleDifferences}. To do so we first require the following \Cref{lemma:opuesto}.

\begin{lemma}\label{lemma:opuesto}
  Let $\alpha_{0}, \dots, \alpha_{n-1}$ be a twofold set of $n$\=/th roots of $a$ with invertible differences.

  Then $\alpha_{\overline{\imath}} = -\alpha_{i}$.
\end{lemma}

\begin{proof}
  Since the set satisfies \Cref{cond:twofold} we know ${\alpha_{i}}^{2} = {\alpha_{\overline{\imath}}}^{2}$. That is
  \[0 = {\alpha_{i}}^{2} - {\alpha_{\overline{\imath}}}^{2} = (\alpha_{i} - \alpha_{\overline{\imath}})(\alpha_{i} + \alpha_{\overline{\imath}}).\]
  Using that the elements have invertible differences we obtain $\alpha_{i} + \alpha_{\overline{\imath}} = 0$.
\end{proof}

\begin{proposition}[Squares of a set satisfying \Cref{cond:evalAtRoots,cond:twofold,cond:invertibleDifferences} also satisfy \Cref{cond:evalAtRoots,cond:twofold,cond:invertibleDifferences}]\label{prop:squaring}
  Let $\alpha_{0}, \dots, \alpha_{n-1}$ be a twofold set of $n$\=/th roots of $a$ with invertible differences.

  Then the set ${\alpha_{0}}^{2}, \dots, {\alpha_{n/2-1}}^{2}$ is a twofold set of $n/2$\=/th roots of $a$ with invertible differences.
\end{proposition}

\begin{proof}
  From the definition of a twofold set it directly follows that the set of squares ${\alpha_{0}}^{2}, \dots,\allowbreak {\alpha_{n/2-1}}^{2}$ is a twofold set of $n/2$\=/th roots of $a$.

  We only need to check if the differences among the squares are still invertible.
  \[ {\alpha_{i}}^{2} - {\alpha_{j}}^{2} = (\alpha_{i} - \alpha_{j})(\alpha_{i} + \alpha_{j}) = (\alpha_{i} - \alpha_{j})(\alpha_{i} - \alpha_{\overline{\jmath}}).\]

  Using \Cref{lemma:opuesto} we have seen the differences among the squares are products of differences among original elements, invertible by hypothesis, implying the squares also satisfy the conditions.
\end{proof}

After these preliminaries one can define the anti-transform from a constructive point of view by reversing the transform algorithm or explicitating the inverse of the transform matrix. However, to get a deeper insight we are going to describe it using the language of Lagrange interpolation.

Given an indexed set of points ${\{\alpha_{i}\}}_{i=0}^{n-1}$ with invertible differences and the evaluations of a polynomial in such points ${\{g(\alpha_{i})\}}_{i=0}^{n-1}$ we can recover the original polynomial $g(x)$ using Lagrange polynomials $l_{i}^{{\{\alpha_{j}\}}_{j=0}^{n-1}}(x)$ as
\[g(x) = \sum_{i=0}^{n-1}g(\alpha_{i})l_{i}^{{\{\alpha_{j}\}}_{j=0}^{n-1}}(x),\quad l_{i}^{{\{\alpha_{j}\}}_{j=0}^{n-1}}(x) = \prod_{\substack{j=0\\j\neq i}}^{n-1}\frac{x-\alpha_{j}}{\alpha_{i}-\alpha_{j}}.\]

The key point is to observe that, due to the particular requirements of our set of evaluation points, that is \Cref{cond:invertibleDifferences,cond:evalAtRoots,cond:twofold}, our Lagrange polynomials factorize in a special way.

Using \Cref{lemma:opuesto} we can see how the Lagrange polynomial splits.
\begin{align*}
  l_{i}^{{\{\alpha_{j}\}}_{j=0}^{n-1}}(x)
& = \prod_{\substack{j=0\\j\neq i}}^{n-1}\frac{x-\alpha_{j}}{\alpha_{i}-\alpha_{j}} \\
& = \frac{x-\alpha_{\overline{\imath}}}{\alpha_{i}-\alpha_{\overline{\imath}}}\prod_{\substack{j=0\\j\not\equiv i \\ \pmod{n/2}}}^{n/2-1}\left(\frac{x-\alpha_{j}}{\alpha_{i}-\alpha_{j}}\right)\left(\frac{x-\alpha_{\overline{\jmath}}}{\alpha_{i}-\alpha_{\overline{\jmath}}}\right) \\
& = \frac{x-\alpha_{\overline{\imath}}}{\alpha_{i}-\alpha_{\overline{\imath}}}\prod_{\substack{j=0\\j\not\equiv i \\ \pmod{n/2}}}^{n/2-1} \frac{x^{2}-{\alpha_{j}}^{2}}{{\alpha_{i}}^{2}-{\alpha_{j}}^{2}} \\
& = l_{i}^{\{\alpha_{i},\alpha_{\overline{\imath}}\}}(x)  l_{i \rem n/2}^{{\{{\alpha_{j}}^{2}\}}_{j=0}^{n/2-1}}(x^{2}).
\end{align*}

It is crucial to note that with a twofold set $l_{\overline{\imath} \rem n/2}^{{\{{\alpha_{j}}^{2}\}}_{j=0}^{n/2-1}}(x) = l_{i \rem n/2}^{{\{{\alpha_{j}}^{2}\}}_{j=0}^{n/2-1}}(x)$.

Then we can write
\begin{align*}g(x) &=
\sum_{i=0}^{n-1}g(\alpha_{i})l_{i}^{{\{\alpha_{j}\}}_{j=0}^{n-1}}(x) \\
 &= \sum_{i=0}^{n-1}g(\alpha_{i})l_{i}^{\{\alpha_{i},\alpha_{\overline{\imath}}\}}(x)  l_{i\rem n/2}^{{\{{\alpha_{j}}^{2}\}}_{j=0}^{n/2-1}}(x) \\
 &=\sum_{i=0}^{n/2-1}\left( g(\alpha_{i})l_{i}^{\{\alpha_{i},\alpha_{\overline{\imath}}\}}(x)+g(\alpha_{\overline{\imath}})l_{\overline{\imath}}^{\{\alpha_{i},\alpha_{\overline{\imath}}\}}(x) \right)
 l_{i}^{{\{{\alpha_{j}}^{2}\}}_{j=0}^{n/2-1}}(x^{2}).
\end{align*}

Once we have this decomposition the advantage of this language is that it allows us to interpret it. We were considering $g(x)$ as $g_{0}(x^{2}) + x g_{1}(x^{2})$.
Polynomials $l_{i}^{{\{{\alpha_{j}}^{2}\}}_{j=0}^{n/2-1}}(x)$ would help us interpolate $g_{0}(x)$ and $g_{1}(x)$ if we had their images $\{g_{0}({\alpha_{i}}^{2})\}$ and $\{g_{1}({\alpha_{i}}^{2})\}$. However the images we have are $\{g(\alpha_{i})\}$.
But then $l_{i}^{\{\alpha_{i},\alpha_{\overline{\imath}}\}}(x)$ and $l_{\overline{\imath}}^{\{\alpha_{i},\alpha_{\overline{\imath}}\}}(x)$ are precisely the polynomials that interpolate $g_{0}({\alpha_{i}}^{2}) + xg_{1}({\alpha_{i}}^{2})$ (a polynomial of degree $1$ that has the desired evaluations as coefficients) from $g(\alpha_{i})$ and $g(\alpha_{\overline{\imath}})$.

As we are going to use it later lets explicitate
\[g(\alpha_{i})l_{i}^{\{\alpha_{i},\alpha_{\overline{\imath}}\}}(x)+g(\alpha_{\overline{\imath}})l_{\overline{\imath}}^{\{\alpha_{i},\alpha_{\overline{\imath}}\}}(x)
= \alpha_{i}\left(\frac{g(\alpha_{i})+g(\alpha_{\overline{\imath}})}{\alpha_{i}- \alpha_{\overline{\imath}}} \right)
+ x \left(\frac{g(\alpha_{i})-g(\alpha_{\overline{\imath}})}{\alpha_{i}- \alpha_{\overline{\imath}}} \right).\]

This can be used to build an efficient interpolation algorithm. From the $n$ evaluations of a polynomial $g$ of degree bounded by $n$ at points ${\{\alpha_{i}\}}_{i=0}^{n-1}$ we can recover the evaluations of polynomials $g_{0}$ and $g_{1}$ at points ${\{{\alpha_{i}}^{2}\}}_{i=0}^{n/2-1}$ (this is done with interpolations of polynomials of degree bounded by $2$, so each requires a constant time and we need a total of $\mathcal{O}(n)$ operations).
Then we use them to interpolate $g_{0}$ and $g_{1}$, each of them polynomials of degree bounded by $n/2$ belonging to $R[x]/\left<x^{n/2}-a\right>$. That is $T(n)= 2T(n/2) + \mathcal{O}(n)$, and we end up again achieving $T(n) = \bigO{n\log{n}}$.

\begin{algorithm}[H]
  \linespread{1.35}\selectfont
   \caption{\textsc{IFFT}}\label{alg:ifft}
\SetAlgoLined{}
\KwInput{A vector of evaluations $\vec{y}$ of size $n$ and a twofold set $\alpha_{0},\dots,\alpha_{n-1}$ of $n$\=/th roots of $a$ with invertible differences}
\KwResult{Coefficients of a polynomial $g(x)$ interpolating $\vec{y}$ at $\alpha_{0},\dots,\alpha_{n-1}$}
\nl{} \lIf{n = 1}{\Return{$y$}}
\nl{} \For{$i \gets 0$ \KwTo{} $n/2-1$}{
  $\vec{y}_{0}[i] \coloneqq \alpha_{i}\left(\frac{\vec{y}[i]+\vec{y}[i+n/2]}{\alpha_{i}- \alpha_{i+n/2}} \right)$\\
  $\vec{y}_{1}[i] \coloneqq \left(\frac{\vec{y}[i]-\vec{y}[i+n/2]}{\alpha_{i}- \alpha_{i+n/2}} \right)$
}
\nl{} $g_{0}(x) \coloneqq \textsc{IFFT}(\vec{y}_{0},{\alpha_{0}}^{2},\dots, {\alpha_{n/2-1}}^{2})$\\
\nl{} $g_{1}(x) \coloneqq \textsc{IFFT}(\vec{y}_{1},{\alpha_{0}}^{2},\dots, {\alpha_{n/2-1}}^{2} )$\\
\nl{} \Return{$g_{0}(x^{2})+xg_{1}(x^{2})$}
\end{algorithm}

\subsection{Efficient multiplication algorithm in \texorpdfstring{$R[x]/\left< x^n - a\right>$}{R[x]/<xⁿ-a>}}

Combining both \Cref{alg:fftRootsOfA,alg:ifft} we describe in \Cref{alg:efficientMultiplication} an efficient multiplication algorithm in $R[x]/\left< x^n - a \right>$.

\begin{algorithm}[H]
  \linespread{1.35}\selectfont
   \caption{\textsc{Efficient FFT Multiplication}}\label{alg:efficientMultiplication}
\SetAlgoLined{}
\KwInput{Two polynomials $g(x)$, $h(x)$ of degree bounded by $n$}
\KwAuxiliary{A twofold set $\alpha_{0},\dots,\alpha_{n-1}$ of $n$\=/th roots of $a$ with invertible differences}
\KwResult{The product $(g \cdot h)(x)$ of $g(x)$ and $h(x)$ in $R[x]/\left< x^n - a \right>$}
\nl{} $\vec{g} \coloneqq \textsc{FFT}(g(x),\alpha_{0}, \dots, \alpha_{n-1})$\\
\nl{} $\vec{h} \coloneqq \textsc{FFT}(h(x),\alpha_{0}, \dots, \alpha_{n-1})$\\
\nl{} $\vec{f} \coloneqq \vec{g} \odot \vec{h}$\\
\nl{} $f(x) \coloneqq \textsc{IFFT}(\vec{f},\alpha_{0}, \dots, \alpha_{n-1})$\\
\Return{$f(x)$}
\end{algorithm}

Thereby our work is to study the existence of sets of evaluation points satisfying \Cref{cond:twofold,cond:invertibleDifferences,cond:evalAtRoots} and how to find them in our desired $R$.

%% file: sections/characterizationOfPoints.tex

\section{Characterization of suitable sets of evaluation points in the ring \texorpdfstring{$\Z_{m}[x]/\left<x^n-a\right>$}{ℤₘ[x]/<xⁿ-a>}}\label{sec:suitableSets}

In this section we focus on $\Z_{m}$ and study the relations among the given conditions to see that these are precisely the required notions and provide necessary and sufficient conditions for the existence of proper evaluation sets.

We can start certifying that \Cref{cond:twofold,cond:evalAtRoots,cond:invertibleDifferences} are indeed independent in general.

\begin{proposition}[\Cref{cond:twofold} does not imply \Cref{cond:invertibleDifferences} in $\Z_{m}$ if $m$ is not a power of a prime]

  Let $\alpha_{0}, \dots, \alpha_{n-1}$ be a twofold set of $n$\=/th roots of $a$ in $\Z_{m}$, where $m$ is not a power of a prime. There is a twofold set of $n$\=/th roots of $a$ in $\Z_{m}$ without invertible differences.
\end{proposition}

\begin{proof}
  Decomposing $m=pq$ with $p$ and $q$ coprime proper factors we can always construct another twofold set defining $\alpha'_{i} \equiv \alpha_{i} \pmod{p}$ but $\alpha'_{i} \equiv \alpha_{0} \pmod{q}$. It would be a twofold set but none of their differences would be invertible.
\end{proof}

When the modulus is a power of a prime ${p^e}$ invertibility comes from being different modulo $p$, which is not implied in general by being different modulo $p^{e}$. We have first to further characterize $n$\=/th roots in $\Z_{p^e}$ when $p$ is prime to address this particular case.

\begin{theorem}[Hensel's lemma as in Theorem~2.23 from~\cite{niven}]\label{thm:hensel}

  Suppose that $f(x)$ is a polynomial with integral coefficients. If $f(x_{0}) \equiv 0 \pmod{p^e}$ and $f'(x_{0}) \not\equiv 0 \pmod{p}$, then there is a unique $t \pmod{p}$ such that $f(x_{0} + tp^e) \equiv 0 \pmod{p^{e+1}}$.
\end{theorem}

In our case we are particularly interested in $f(x) = x^n-a$, with $n$ a power of $2$, so $f'(x)=nx^{n-1}$. Since our solutions are $n$\=/th roots of $a$ (when considered modulo $p^{e}$ and therefore also modulo $p$) $n {x_{0}}^{n-1} \not\equiv 0 \pmod{p}$ as long as $p \neq 2$ and $a \neq 0$.

\begin{corollary}\label{cor:rootsCorrespondence}
  There is a one to one correspondence of $n$\=/th roots in $\Z_{p^e}$ and in $\Z_{p}$, where $p$ is an odd prime.
\end{corollary}

\begin{proof}
  On the one hand we can see each root in $\Z_{p^e}$ as a root in $\Z_{p}$ by applying $\prem{p}$.

  On the other hand one just needs to apply \Cref{thm:hensel} iteratively $e-1$ times from $\Z_{p}$ to $\Z_{p^e}$.
\end{proof}

In particular this implies the order, as elements of the group, is preserved given that powers of a root are uniquely lifted to powers of the lifted root.

\begin{remark}
  We omit here the case $p=2$ as in the following sections we are going to see that other conditions forbid this particular case.
\end{remark}

\begin{proposition}[\Cref{cond:evalAtRoots} implies \Cref{cond:invertibleDifferences} in $\Z_{m}$ if $m$ is a power of an odd prime]

  Let $m=p^e$, with $p$ an odd prime, and let $\alpha_{0}, \dots, \alpha_{n-1}$ be $n$ different $n$\=/th roots of $a$ in $\Z_{p^e}$. Then $\alpha_{i}-\alpha_{j}$ is invertible in $\Z_{p^e}$ for all $i \neq j$.
\end{proposition}

\begin{proof}
  Every $\alpha_{i}$ is a root of $f(x) = x^n-a$ when considered modulo $p$. By the previous corollary we have $\alpha_{i} \not\equiv \alpha_{j} \pmod{p^e}$ implies $\alpha_{i} \not\equiv \alpha_{j} \pmod{p}$. Therefore $\gcd(\alpha_{i} - \alpha_{j},p)=1$ as we wanted.
\end{proof}

\begin{proposition}[\Cref{cond:twofold} implies \Cref{cond:invertibleDifferences} in $\Z_{m}$ if $m$ is a power of an odd prime]
  Let $m=p^e$, with $p$ an odd prime, and let $\alpha_{0}, \dots, \alpha_{n-1}$ be a twofold set of $n$\=/th roots of $a$ in $\Z_{p^e}$. Then $\alpha_{i}-\alpha_{j}$ is invertible in $\Z_{p^e}$ for all $i \neq j$.
\end{proposition}

\begin{proof}
  \Cref{cond:twofold} implies \Cref{cond:evalAtRoots} and \Cref{cond:evalAtRoots} implies \Cref{cond:invertibleDifferences}.
\end{proof}

Working with an arbitrary modulus $m$, that has $m={p_{1}}^{e_1} \dots {p_{k}}^{e_k}$ as its prime decomposition, we can completely determine any $d \in \Z_{m}$, via the Chinese Remainder Theorem (CRT), from $d_{(i)}$ such that
\begin{align*}
  d_{(1)} & \equiv d \pmod{{p_{1}}^{e_1}}\\
  d_{(2)} & \equiv d \pmod{{p_{2}}^{e_2}}\\
  & \vdotswithin{\equiv} \\
  d_{(k)} & \equiv d \pmod{{p_{k}}^{e_k}}.
\end{align*}

Using this representation we can prove the following theorem.

\begin{theorem}[\Cref{cond:evalAtRoots,cond:invertibleDifferences} hold if and only if \Cref{cond:evalAtRoots} holds modulo every ${p_{j}}^{e_j}$]

Let $\alpha_{0}, \dots, \alpha_{n-1} \in \Z_{m}$ be a set of different $n$\=/th roots of $a\in \Z_{m}$, and let $m={p_{1}}^{e_1} \dots {p_{k}}^{e_k}$ be the prime decomposition of an odd module $m$.

The differences among these elements are invertible modulo $m$ if and only if all the elements are still different when considered modulo any of the ${p_{j}}^{e_j}$.
\end{theorem}

\begin{proof}
  Follows the same ideas as the previous propositions, since an element is invertible if and only if it is invertible modulo all the coprime factors of a factorization of its modulus and we have seen (\Cref{cor:rootsCorrespondence}) that such roots are different modulo ${p_{j}}^{e_j}$ if and only if they are different modulo $p_{j}$, and therefore have invertible differences.
\end{proof}

The same way we can see how other implications are not true in general either.

\begin{proposition}[\Cref{cond:invertibleDifferences,cond:evalAtRoots} do not imply \Cref{cond:twofold} in $\Z_{m}$ if $m$ is not a power of a prime]

  Not every $n$\=/set of $n$\=/th roots of $a \in \Z_{m}$ with invertible differences is a twofold set of $n$\=/th roots. It is not the case in general, not even for roots of unity.
\end{proposition}

\begin{proof}
  Consider the following set of $4$\=/th roots of unity in $\Z_{65}$:
\begin{align*}
  12^4 &\equiv 1 &\pmod{65} &&\qquad\qquad\qquad\qquad 12^2 &\equiv 14 &\pmod{65}\\
  14^4 &\equiv 1 &\pmod{65} && 14^2 &\equiv 1  &\pmod{65}\\
  18^4 &\equiv 1 &\pmod{65} && 18^2 &\equiv 64  &\pmod{65}\\
  21^4 &\equiv 1 &\pmod{65} && 21^2 &\equiv 51  &\pmod{65}
\end{align*}

And their set of differences is $\{2,3,4,6,7,9\}$, all of them invertible elements modulo $65$.
\end{proof}

Observe that the evaluation points from the example are indeed, after some reorderings, a twofold set of $n$\=/th roots of unity in $\Z_{5}$ and in $\Z_{13}$, but the \emph{reorderings} are different.
\begin{align*}
  12^2 &\equiv 18^2 &\pmod{5} &&\qquad\qquad\qquad\qquad 12^2 &\equiv 14^2 &\pmod{13} \\
  14^2 &\equiv 21^2 &\pmod{5} && 18^2 &\equiv 21^2 &\pmod{13}
\end{align*}

Once again the proposition does hold if we work modulo a power of an odd prime. To prove it we require a couple of lemmas that show the important role of roots of unity and allow us to focus on them with the goal of better understanding these orderings.

\begin{lemma}[\Cref{cond:evalAtRoots,cond:invertibleDifferences} imply the evaluation points are invertible]\label{lemma:invertibilityEvaluationPoints}

  Let $\alpha_{0}, \dots, \alpha_{n-1}$ be different $n$\=/th roots of $a$ in $\Z_{m}$ such that $\alpha_{i} - \alpha_{j}$ is invertible in $\Z_{m}$ for all $i \neq j$. Then $\alpha_{i}$ is invertible in $\Z_{m}$ for all $i$.

\end{lemma}

\begin{proof}
  Choose two indices $i$ and $j$ and let $d$ be a square-free common divisor of $\alpha_{i}$ and $m$.
  Since ${\alpha_{i}}^n \equiv a \pmod{m}$ and ${\alpha_{j}}^n \equiv a \pmod{m}$ we have that $m \vert {\alpha_{i}}^n - {\alpha_{j}}^n$, implying that $d \vert \alpha_{j}$ (here we have to use that $d$ is square free).
  We would finally get $d \vert \alpha_{i} - \alpha_{j}$ but we know $\gcd(\alpha_{i} - \alpha_{j},m) = 1$ and therefore $d=1$, proving $\gcd(\alpha_{i},m) = 1$ for all $i$.
\end{proof}

\begin{remark}
  If roots of $a$ are invertible it directly follows that $a$ itself has to be invertible. We will impose it when required since this argument implies it is a necessary condition for the existence of the inverse transform $T^{-1}$.
\end{remark}

\begin{lemma}[Roots of $a$ and roots of $1$]\label{lemma:alphasAndOmegasOne}

  Let $a \in \Z_{m}$. The following two statements are equivalent:
  \begin{enumerate}[(i)]
    \item\label{enum:rootA} The set $\alpha_{0}, \dots, \alpha_{n-1} \in \Z_{m}$ satisfies \Cref{cond:evalAtRoots,cond:invertibleDifferences}.
    \item\label{enum:rootOne} The set $\alpha_{0}, \dots, \alpha_{n-1} \in \Z_{m}$ can be constructed from an invertible $n$\=/th root of $a$, lets denote it $\alpha$, and $n$ different $n$\=/th roots of unity $\omega_{0}, \dots, \omega_{n-1}$ with invertible differences in $\Z_{m}$ such that $\alpha_{i} = \alpha\omega_{i}$.
  \end{enumerate}
\end{lemma}

\begin{proof}
  Lets prove both implications:
  \begin{itemize}
    \item(\ref{enum:rootA})$\implies$(\ref{enum:rootOne})

    Let $\alpha_{0},\dots,\alpha_{n-1}$ be the roots satisfying the conditions.

    We can define $\alpha=\alpha_{0}$ (invertible by \Cref{lemma:invertibilityEvaluationPoints}) and $\omega_{i} = \alpha_{i} \cdot {\alpha_{0}}^{-1}$. We can check $\omega_{i}$ are roots of unity and their differences are invertible
    \[{(\omega_{i} - \omega_{j})}^{-1} = {\left(\frac{\alpha_{i}}{\alpha_{0}} - \frac{\alpha_{j}}{\alpha_{0}}\right)}^{-1} = \alpha_{0} \cdot {(\alpha_{i}-\alpha_{j})}^{-1}.\]

    \item(\ref{enum:rootOne})$\implies$(\ref{enum:rootA})

    Let $\alpha, \omega_{0}, \dots, \omega_{n-1}$ be a set of roots satisfying the conditions.

    Define now $\alpha_{i} = \alpha\cdot\omega_{i}$. Again by construction all $\alpha_{i}$ are $n$\=/th roots of $a$ and their differences are invertible since
    \[{(\alpha_{i} - \alpha_{j})}^{-1} = {(\alpha\cdot\omega_{i} - \alpha\cdot\omega_{j})}^{-1} = \alpha^{-1}{(\omega_{i} - \omega_{j})}^{-1}.\]
  \end{itemize}
\end{proof}

This motivates the definition of a sufficient condition that, as we are going to see, will be necessary when $m$ is a power of an odd prime.

\begin{definition}[$(\alpha,\omega)$\=/set]
  Let $\alpha$ be any $n$\=/th root of an invertible $a \in R$ and $\omega$ an $n$\=/th root of unity in $R$ of order $n$ whose powers have invertible differences. Then the set defined as $\alpha_{i} = \alpha\omega^{i}$ is said to be an  \emph{$(\alpha,\omega)$\=/set}.
\end{definition}

\begin{condition}[$(\alpha,\omega)$\=/set]\label{cond:alphaOmega}
We say a set of evaluation points satisfies \Cref{cond:alphaOmega} if (after some reordering) it is an $(\alpha,\omega)$\=/set for some $\alpha,\omega \in R$.
\end{condition}

\begin{remark}
  Let $m={p_{1}}^{e_1} \dots {p_{k}}^{e_k}$ be the prime decomposition of the module $m$.
  An $n$\=/th root of unity $\omega \in \Z_{m}$ can be determined from $\omega_{(i)} \equiv \omega \pmod{p_{i}^{e_{i}}}$ and the order of $\omega$ is just the least common multiple of the orders of $\omega_{(i)}$ in $\Z_{p_{i}^{e_{i}}}$.
  Since we choose $n$ to be a power of $2$ and all the orders of each $\omega_{(i)}$ divide $n$ the least common multiple is just going to be the maximum of the orders. Then, for $\omega$ to be an $n$\=/th root of unity of order $n$ it is only necessary that one of these $\omega_{(j)}$ has order $n$. However as we also need to impose invertibility of the differences of its powers then every $\omega_{(i)}$ has to have order $n$.
\end{remark}

\begin{proposition}[\Cref{cond:alphaOmega} implies \Cref{cond:twofold,cond:evalAtRoots,cond:invertibleDifferences}]\label{prop:alphaOmegaSuficiente}

  An $(\alpha,\omega)$\=/set in $\Z_{m}$ is a twofold set of $n$\=/th roots of $a$ with invertible differences.
\end{proposition}

\begin{proof}
  Let $\alpha_{i} = \alpha\omega^i$ and let $m={p_{1}}^{e_1}{p_{2}}^{e_2} \dots {p_{k}}^{e_k}$ be the prime decomposition of $m$.

  We can start checking \Cref{cond:invertibleDifferences}. From the proof of \Cref{lemma:alphasAndOmegasOne} we know an $(\alpha,\omega)$\=/set of $n$\=/th roots satisfies \Cref{cond:invertibleDifferences} if $\alpha$ is invertible and $\omega^i - \omega^j$ are invertible too.

  Since $a$ is invertible in $\Z_{m}$ then $\alpha$ as an $n$\=/th root of $a$ has to be invertible too. Otherwise if $\alpha \equiv 0 \pmod{p_{j}}$ for some $j$ then $a \equiv 0 \pmod{p_{j}}$ for the same $j$ and it would not be invertible either, contradicting the statement.

  The second condition is ensured from the definition. Then every difference is invertible modulo $m$.

  \Cref{cond:evalAtRoots} follows from the construction of an $(\alpha,\omega)$\=/set.

  Then for \Cref{cond:twofold} let $i \equiv j \pmod{2^{{\log(n)} - k}}$ and, without loss of generality, assume $i \geq j$ and therefore $i = j + c\cdot2^{{\log(n)}-k}$ for some non-negative integer $c$. Then
  \[ \alpha_{i}^{2^{k}} =
  {(\alpha_{0}\omega^{i})}^{2^{k}} =
  \alpha_{0}^{2^{k}}\omega^{(j + c\cdot2^{{\log(n)}-k})\cdot2^{k}} =
  \alpha_{0}^{2^{k}}\omega^{j \cdot 2^{k} + c \cdot n} =
  {(\alpha_{0}\omega^{j})}^{2^{k}} =
  \alpha_{j}^{2^{k}}.  \]
\end{proof}

This condition is quite convenient, for example it allows us to explicitly describe the transform from its Vandermonde matrix, as we mentioned before, in a compact way as ${(V)}_{i\,j} = {(\alpha \omega^{i})}^{j}$ and ${(V^{-1})}_{i\,j} =
\frac{1}{n}{(\alpha^{-1}\omega^{-j})}^{i}$. Notice that the inverse matrix looks like $1/n$ times the transpose of the evaluation matrix of an $(\alpha^{-1},\omega^{-1})$\=/set. If we were working in $\Z_{m}[x]/\left<x^n-1\right>$ with $\alpha=1$ the transpose would be irrelevant as the matrix would be symmetric and the same efficient recursive evaluation techniques from the direct transform would work directly for its inverse.

However even if \Cref{cond:alphaOmega} implies \Cref{cond:twofold,cond:evalAtRoots,cond:invertibleDifferences} the converse is not true in general. Besides that we can see it holds for some particular cases, when $m$ is a power of a prime.

\begin{proposition}[\Cref{cond:invertibleDifferences,cond:evalAtRoots} do imply \Cref{cond:alphaOmega} in $\Z_{m}$ if $m$ is a power of an odd prime]\label{prop:alphaOmegaPotenciaDePrimo}

  Let $\alpha_{0},\dots,\alpha_{n-1}$ be a set of $n$\=/th roots of $a$ with invertible differences in $\Z_{p^e}$, with $p$ an odd prime. This is (except reordering) an $(\alpha,\omega)$\=/set of $n$\=/th roots of $a$ in $\Z_{p^e}$.
\end{proposition}

\begin{proof}
  From \Cref{lemma:alphasAndOmegasOne} we deduce the existence of an invertible $\alpha$ and $\omega_{0}, \dots, \allowbreak \omega_{n-1}$ with invertible differences such that $\alpha_{i} = \alpha\cdot\omega_{i}$.

  As $\Z_{p}$ is a field its multiplicative group is a cyclic group of order $p-1$. 
  Then the set of $n$\=/th roots of unity in $\Z_{p}$ is also a group, and as a subgroup of a cyclic group it is also cyclic.

  It is also known that $x^n-1$ has at most $n$ solutions modulo $p$ (Theorem~2.6 from~\cite{niven}), therefore the $n$ roots $\omega_{i} \prem{p}$ (all different since $\alpha\omega_{i} \prem{p}$ are lifted to different points in $\Z_{p^e}$) form the whole cyclic group of roots of unity and in consequence are generated by one of them.

  As there is a one to one correspondence among $n$\=/th roots in $\Z_{p}$ and in $\Z_{p^e}$ the original $\omega_{i} \in \Z_{p^e}$ are generated too by one $\omega$ of order $n$. That is, there exists a permutation $\pi$ such that $\omega_{\pi(i)} = \omega^i$.

  After this reordering defined by $\pi$ the set $\alpha_{\pi(0)}, \alpha_{\pi(1)}, \dots, \alpha_{\pi(n-1)}$ is an $(\alpha,\omega)$\=/set.
\end{proof}

\begin{proposition}[\Cref{cond:invertibleDifferences,cond:evalAtRoots} do imply \Cref{cond:twofold} in $\Z_{m}$ if $m$ is a power of an odd prime]\label{prop:twofoldPotenciaDePrimo}

  Let $\alpha_{0},\dots,\alpha_{n-1}$ be a set of $n$\=/th roots of $a$ with invertible differences in $\Z_{p^e}$, with $p$ an odd prime. This is (except reordering) a twofold\=/set of $n$\=/th roots of $a$ in $\Z_{p^e}$.
\end{proposition}

\begin{proof}
  Direct as we know by \Cref{prop:alphaOmegaPotenciaDePrimo} that \Cref{cond:invertibleDifferences,cond:evalAtRoots} imply \Cref{cond:alphaOmega} in $\Z_{m}$ if $m$ is a power of an odd prime and \Cref{cond:alphaOmega} always implies \Cref{cond:twofold}, as seen in \Cref{prop:alphaOmegaSuficiente}.
\end{proof}

This new condition seems the right choice, and the FFT is usually introduced from constructions equivalent to this definition, but we should study first if restricting to this particular set of evaluation points reduces the options for computing an FFT multiplication when $m$ is not a power of an odd prime.

Once again \Cref{cond:invertibleDifferences,cond:evalAtRoots} being true modulo every ${p_{j}}^{e_j}$ should imply that \Cref{cond:alphaOmega} holds modulo every ${p_{j}}^{e_j}$, but the permutations might be different.

There are cases where these permutations allow \Cref{cond:twofold} to be true while \Cref{cond:alphaOmega} is not.

\begin{theorem}[\Cref{cond:invertibleDifferences,cond:evalAtRoots,cond:twofold} do not imply \Cref{cond:alphaOmega} in general, but imply the existence of a set satisfying \Cref{cond:alphaOmega}]\label{thm:relationConditionsAlphaOmega}

  Let $\alpha_{0},\dots,\alpha_{n-1}$ be a twofold set of $n$\=/th roots of $a$ with invertible differences in $\Z_{m}$, where $m$ is odd and not a power of a prime and $n$ is a power of two greater than $2^2$.
  Then $a$ has an inverse in $\Z_{m}$, there is both a twofold set $\alpha'_{0},\dots,\alpha'_{n-1}$ of $n$\=/th roots of $a$ in $\Z_{m}$ with invertible differences not satisfying \Cref{cond:alphaOmega} and a set $\alpha''_{i} = \alpha\omega^i$ where $\alpha$ is an $n$\=/th root of $a$ and $\omega$ is a root of unity of order $n$ with all its powers having invertible differences, that is, an $(\alpha,\omega)$\=/set.
\end{theorem}

\begin{proof}
  Following \Cref{lemma:invertibilityEvaluationPoints} we get $a$ is invertible.

  Let $m={p_{1}}^{e_1}{p_{2}}^{e_2} \dots {p_{k}}^{e_k}$ be the prime decomposition of $m$.

  There is a unique (except reordering) twofold set of $n$\=/th roots of $a$ with invertible differences in $\Z_{{p_{j}}^{e_j}}$. This is the case since the values of each of the roots of $a$ modulo ${p_{j}}^{e_j}$ are completely determined, being the unique elements lifted to $\Z_{{p_{j}}^{e_j}}$ from the unique $n$ roots of $a$ in $\Z_{p_{j}}$. By \Cref{prop:alphaOmegaPotenciaDePrimo} it satisfies \Cref{cond:alphaOmega}. The only thing we could choose is the respective order they have.

  This order is irrelevant in $\Z_{{p_{j}}^{e_{j}}}$ but becomes important when considering $\Z_{m}$ as once we fix an order modulo ${p_{1}}^{e_1}$ the different respective orders for the remaining ${p_{j}}^{e_{j}}$ would produce different elements in $\Z_{m}$.

  From the twofold set definition we got a tree structure in \Cref{fig:binarytree} and a specific notation for the points (a bit decomposition of the index). The twofold structure is only preserved by the tree structure, so the only possible reorderings are those that come from swapping left and right children of a node. That is, choosing $b_{i-1}\dots b_{0}$ (an internal node), and mapping $\alpha_{b_{{\log(n)}-1} \dots b_{i} \dots b_{0}}$ to $\alpha_{b_{{\log(n)}-1}\dots\overline{b_{i}} \dots b_{0}}$ for all $b_{{\log(n)}-1},\dots,b_{i+1}$ still preserves this structure.

  For example, choosing nodes $\alpha_{1}$ and $\alpha_{10}$ we obtain a different ordering like in \Cref{fig:binarytreeSwapp}.

  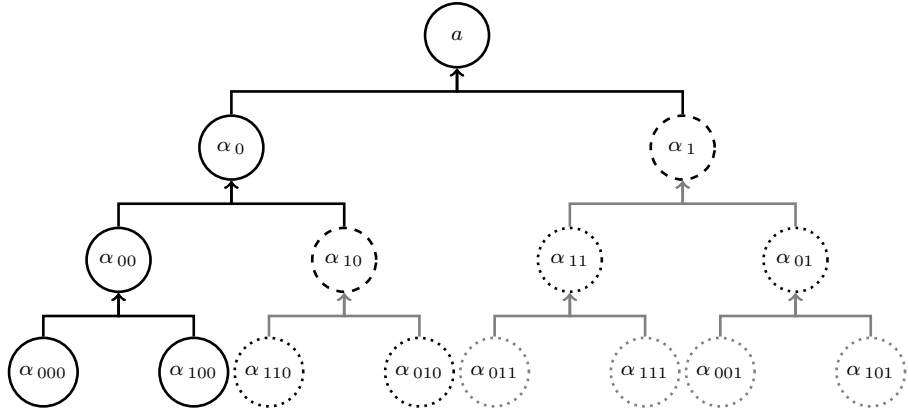
\begin{figure}[ht]
   \begin{minipage}{\textwidth}
     \centering
      \begin{tikzpicture}[every node/.style={minimum size=0.85cm,circle,draw,line width=1pt},edge from parent/.style={draw,line width=1pt,<-},
        level/.style={sibling distance=60mm/#1}, level distance=4.25em,
        edge from parent path={
        (\tikzparentnode) |-   
        ($(\tikzparentnode)!0.5!(\tikzchildnode)$) -| 
        (\tikzchildnode)},
        scale=1,transform shape
        ]
        \footnotesize
        \node [color=black] (a){$a$}
          child {node [color=black] (a0) {$\alpha_{\,0}$}
            child {node [color=black] (a00) {$\alpha_{\,00}$}
              child {node [color=black] (a000) {$\alpha_{\,000}$} edge from parent}
              child {node [color=black] (a100) {$\alpha_{\,100}$}  edge from parent}
            }
            child {node [dashed,text=black] (a10) {$\alpha_{\,10}$}
              child {node [dotted,text=black] (a110) {$\alpha_{\,110}$} edge from parent [color=black!50]}
              child {node [dotted,text=black] (a010) {$\alpha_{\,010}$} edge from parent [color=black!50]}
            }
          }
          child {node [dashed,text=black] (a1) {$\alpha_{\,1}$}
            child {node [dotted,text=black] (a11) {$\alpha_{\,11}$}  edge from parent [color=black!50]
              child {node [dotted,text=black] (a011) {$\alpha_{\,011}$}  edge from parent}
              child {node [dotted,text=black] (a111) {$\alpha_{\,111}$}  edge from parent}
            }
            child {node [dotted,text=black] (a01) {$\alpha_{\,01}$}  edge from parent [color=black!50]
              child {node [dotted,text=black] (a001) {$\alpha_{\,001}$}  edge from parent}
              child {node [dotted,text=black] (a101) {$\alpha_{\,101}$}  edge from parent}
            }
        };
      \end{tikzpicture}
      \caption{Twofold set of $8$\=/th roots swapping left and right descendants of $\alpha_{1}$ and $\alpha_{10}$}\label{fig:binarytreeSwapp}
    \end{minipage}
  \end{figure}

  This means that given a twofold set there are exactly $2^{n-1}$ possible reorderings (since the tree has exactly $n-1$ inner nodes and we can swap or not each of them).

  When odd $m$ is not a power of a prime once we fix the order of the roots modulo ${p_{1}}^{e_1}$ each possible reordering of the roots modulo ${p_{j}}^{e_{j}}$ for the remaining $j$ produces a new twofold set with invertible differences, for a total of $2^{(k-1)(n-1)}$ possibilities.

  On the other hand we can count the number of $(\alpha,\omega)$\=/sets. To do so it is important to notice that $(\alpha,\omega)$ and a different pair $(\beta,\xi)$ can generate the same set (just in a different order). If that is the case let $\pi \in \mathfrak{S}_{n}$ be the permutation such that $\beta_{i} = \alpha_{\pi(i)}$.
  Let $c=\pi(0)$, then $\beta = \beta_{0} = \alpha_{\pi(0)} = \alpha\omega^{c}$.
  Let $d = \pi(1) - \pi(0)$, then $\xi = \beta_{1}/\beta_{0} = \alpha_{\pi(1)}/\alpha_{\pi(0)} = (\alpha\omega^{\pi(1)})/(\alpha\omega^{\pi(0)}) = \omega^{\pi(1)-\pi(0)} = \omega^{d}$. 

  This allows us to completely characterize permutation $\pi$ as
  \[\beta_{i} = \beta\xi^{i} = \alpha\omega^{c}{(\omega^{d})}^{i} = \alpha\omega^{c + di} = \alpha_{c+di}.\]
  For $\pi$ to be a permutation $d$ has to be invertible modulo $n$ so it has to be odd.

  That is everything required as any pair of $c \in \Z_{n}$ and odd $d \in \Z_{n}$ would produce a new set of generators $(\alpha\omega^{c},\omega^{d})$ for the same set. That is, each $(\alpha,\omega)$\=/set with invertible differences can be constructed from $n^2/2$ different pairs of roots of $a$ and $1$ (we have $n$ options for $c$ and $n/2$ options for $d$).

  Then, from the existence of a twofold set with invertible differences and previous propositions we know there are $n$ possible $\alpha_{(j)}$ roots of $a$ in $\Z_{{p_{j}}^{e_{j}}}$ and $n/2$ possible $\omega_{(j)}$ roots of unity in $\Z_{{p_{j}}^{e_{j}}}$ of order $n$ that when combined via the CRT would define an $(\alpha'',\omega'')$\=/set with invertible differences. That is a total of $n^{2k}/2^{k}$ pairs of generators, and since every set is defined by $n^2/2$ pairs we would have $n^{2(k-1)}/2^{(k-1)}$ unique sets.

  However, if $n>2^2$ then $2^{(k-1)(n-1)}>n^{2(k-1)}/2^{(k-1)}$ and therefore some of the $2^{(k-1)(n-1)}$ sets ${\{\alpha'_{i}\}}_{i}$ satisfying \Cref{cond:invertibleDifferences,cond:evalAtRoots,cond:twofold} would not satisfy \Cref{cond:alphaOmega}.

\end{proof}

This important theorem ensures that, even if \Cref{cond:alphaOmega} is not necessary to design an efficient FFT multiplication algorithm, as it is sometimes indirectly taken for granted in the literature, we can safely assume it when working in $\Z_{m}[x]$ as it adds no additional restrictions on $m$ to the necessary \Cref{cond:invertibleDifferences,cond:evalAtRoots,cond:twofold}.
Only having done this analysis we can safely use $(\alpha,\omega)$\=/sets when convenient without loosing any generality.

\subsection{Existence and construction of suitable roots in \texorpdfstring{$\Z_{m}$}{ℤₘ}}\label{sec:existenceAndConstruction}

Once we have established the necessary conditions for the transform to be useful to efficiently compute the product of two polynomials modulo $x^{n} - a$ we are left with the task of studying whether such points exist in our desired ring and how to find them.

As we have seen in \Cref{thm:relationConditionsAlphaOmega} the existence of suitable points satisfying \Cref{cond:evalAtRoots,cond:invertibleDifferences,cond:twofold} implies the existence of an $(\alpha,\omega)$\=/set satisfying \Cref{cond:alphaOmega}.
For convenience we are going to characterize when such set of points exists and how to find it.

\begin{remark}
  When $R=\mathbb{C}$ we just have to choose $\omega = e^{\frac{i2\pi}{n}}$ and $\alpha = \sqrt[n]{a}$ (for example $\alpha=1$ if $a=1$ or $\alpha = e^{\frac{i\pi}{n}}$ if $a=-1$). This choice is the standard Fourier Transform and it is usually introduced directly in the literature.
\end{remark}

Now we can consider the case $R=\Z_{m}$.

On the one hand finding a primitive $n$\=/th root of unity $\omega$ in $\Z_{m}$, that is, an  $n$\=/th root of unity of order $n$, such that all its powers have invertible differences implies finding a primitive $n$\=/th root of unity $\omega_{(i)}$ in every $\Z_{{p_{i}}^{e_{i}}}$ and then reconstruct $\omega$ using the CRT\@.

That way we have reduced the problem of finding a primitive $n$\=/th root of unity in $\Z_{m}$ for an arbitrary $m$ to finding a primitive $n$\=/th root of unity in $\Z_{p^e}$.

The proof of \Cref{thm:hensel} in~\cite{niven} explicitly tells us how to lift a solution $x_{e}$ modulo $p^{e}$ to a solution $x_{e+1}$ modulo $p^{e+1}$. It can be computed recursively using $x_{e+1} \equiv x_{e}-f(x_{e})\overline{f'(x_{1})} \pmod{p^{e+1}}$ where $\overline{f'(x_{1})}$ is the inverse of $f'(x_{1})$ when considering it in $\Z_{p}$. Recall $f(x)$ was $x^{n} - a$ and therefore $f'(x) = nx^{n-1}$.

The only step of this computation that is not immediate is to compute $\overline{f'(x_{1})}$. Since we are sure $f'(x_{1})$ is not $0$ modulo $p$ we can use the Extended Euclidean algorithm to compute integers $r$ and $s$ so that $f'(x_{1})r + ps = \gcd(f'(x_{1}),p) = 1$ and $r$ would be the desired $\overline{f'(x_{1})}$.

We finally want to analyze under which conditions on $p$ and $n$ do primitive $n$\=/th roots of unity exist in $\Z_{p}$ and how to find them. We obtain necessary conditions from Fermat's Little \Cref{thm:fermat}.

\begin{theorem}[Fermat's Little Theorem as in Theorem~2.7 from~\cite{niven}]\label{thm:fermat}

  Let $p$ be a prime. If $p \not\vert x_{0}$ then
  \[{x_{0}}^{p-1} \equiv 1 \pmod{p}.\]
\end{theorem}

\begin{corollary}
  If $\Z_{p}$ contains an $n$\=/th root of unity of order $n$ then $n \vert p-1$.
\end{corollary}

\begin{proof}
  Let $x_{0}$ be an $n$\=/th root of unity of order $n$ in $\Z_{p}$. By \Cref{thm:fermat} we have ${x_{0}}^{p-1} \equiv 1 \pmod{p}$ and therefore its order divides $p-1$, that is, $n \vert p-1$ or $p = kn+1$ for an integer $k$.
\end{proof}

\begin{remark}
  Taking advantage of this condition, in the particular case $f(x) = x^{n}\pm1$, the calculation of $\overline{f'(x_{1})}$, with $x_{1}$ an $n$\=/th root of $\mp1$ modulo $p$, can now be computed explicitly as $\overline{f'(x_{1})} = \pm x_{1}(p-1)/n$, as can be easily checked computing $\overline{f'(x_{1})} \cdot f'(x_{1})$ and getting $(\pm x_{1}(p-1)/n)(n{x_{1}}^{n-1}) \equiv 1 \pmod{p}$.
\end{remark}

\begin{remark}
  This necessary condition rules out all the additional considerations we were having about $n \not\equiv 0 \pmod{p_{i}}$, for example ensuring $m$ has to be odd. Observe that in \Cref{alg:ifft} we have to compute a quotient with $\alpha_{i} - \alpha_{\overline{\imath}}$, that is, $2\alpha_{i}$ in the denominator. Requesting
  \Cref{cond:evalAtRoots,cond:invertibleDifferences,cond:twofold} always implies that twice the unity of the ring has to be invertible (via \Cref{lemma:opuesto,lemma:invertibilityEvaluationPoints}).
\end{remark}

\begin{corollary}
  $\Z_{p}$ contains $n$\=/th roots of unity of order $n$ if and only if $n \vert p-1$.
\end{corollary}

\begin{proof}
  We have already seen one implication, lets consider now $\Z_{p}$ with $n \vert p-1$. We know the multiplicative group $\Z_{p}^{*}$ is cyclic. Take a generator $g$, it has order $p-1$ and since $n \vert p-1$ we can choose $\omega = g^{(p-1)/n}$ and by construction it would be an $n$\=/th root of unity of order $n$.
\end{proof}

\begin{remark}
  Notice that, by Dirichlet's theorem on arithmetic progressions, there are infinitely many primes of this form.
\end{remark}

Now we can start computing $\alpha\pmod{p_{j}}$ and $\omega\pmod{p_{j}}$.
Choosing $u_{j}$ a quadratic nonresidue in $\Z_{p_{j}}$ we can let $\omega_{(j)} \equiv u_{j}^{(p_{j}-1)/n} \in \Z_{p_{j}}$ and then lift it to $\Z_{{p_{j}}^{e_{j}}}$ using again the constructive proof of \Cref{thm:hensel}.

Notice first that there is no known deterministic polynomial-time algorithm able to find a quadratic nonresidue (see~\cite{AlgorithmicNumberTheory}) however checking if a uniformly random element from $\Z^{*}_{p_{j}}$ is a quadratic nonresidue can be done computing the Legendre symbol $\left(\frac{u_{j}}{p_{j}}\right) \equiv {u_{j}}^{(p_{j}-1)/2}$, as it is $-1$ if and only if it is a quadratic nonresidue modulo $p$ (see Theorem~3.1 from~\cite{niven}), and has a success probability of almost one half.

On the one hand we know $n \vert p_{j}-1$ is a condition for the existence of appropriate roots, so $\omega_{(j)}={u_{j}}^{(p_{j}-1)/n}$ is well defined. Its $n$\=/th power is ${\omega_{(j)}}^{n}={u_{j}}^{p_{j}-1}=1$ (once again by~\Cref{thm:fermat}).

The only thing we have left is to check all its powers have invertible differences.
If it was not the case then $\omega_{(j)}$ would have order $k$ with $k<n$. This cannot be possible because since we already know ${\omega_{(j)}}^{n}=1$ then it would imply $k \vert n$, and if $k$ was a power of $2$ strictly smaller than $n$ we would get a contradiction as, by construction, $-1 = u_{j}^{(p_{j}-1)/2} = {\omega_{(j)}}^{n/2} = {\left({\omega_{(j)}}^{k}\right)}^{n/(2k)} = 1^{n/(2k)} = 1$.

This ensures our final $\omega$ meets the required conditions.

For computing $\alpha$ we might have different approaches. If $a=1$ the trivial solution $\alpha=1$ works perfectly fine for our purposes. If $a=-1$ we can follow an analogous procedure, noticing that as $n$\=/th roots of $-1$ are $2n$\=/th roots of $1$ the condition is now that $2n \vert p_{j}-1$, and let $\alpha_{(j)} \equiv {u_{j}}^{(p_{j}-1)/2n} \pmod{p_{j}}$, once again lifting them and computing the final $\alpha$ from its CRT representation.

If $a \neq \pm 1$ then our alternative would be to make use of the Tonelli–Shanks Algorithm~\cite{shanks1973five}.

\begin{algorithm}[H]
    \linespread{1.35}\selectfont
     \caption{\textsc{Tonelli–Shanks} (\cite{shanks1973five})}\label{alg:TonelliShanks}
  \SetAlgoLined{}
  \KwInput{An odd prime $p$, a quadratic residue $a\in \Z_{p}$ and a nonresidue $u\in\Z_{p}$}
  \KwResult{An element $\alpha\in\Z_{p}$ square root of $a$}
  \nl Let $v$ and $s$ be such that $p-1=v2^{s}$ and $v$ is odd.\\
  \nl Let\\ \Indp
  $\begin{aligned}
    k &\leftarrow s \\
    c &\leftarrow u^{v} \\
    t &\leftarrow a^{v} \\
    r &\leftarrow a^{\frac{v+1}{2}}
  \end{aligned}$ \\ \Indm
  \nl \While{$t \neq 0 \land t\neq 1$}{
    Find least $i$, $0<i<k$, such that $t^{2^{i}}=1$\\
    Let $d \leftarrow c^{2^{k-i-1}}$ and set\\ \Indp
    $\begin{aligned}
       k &\leftarrow i \\
       c &\leftarrow d^{2} \\
       t &\leftarrow td^{2} \\
       r &\leftarrow rd
     \end{aligned}$ \\ \Indm
  }
  \nl \If{$t=0$}{
    \Return{$\alpha = 0$}
  }
  \nl \If{$t=1$}{
    \Return{$\alpha = r$}
  }
\end{algorithm}

This algorithm allows us to compute a square root of a quadratic residue in $\Z_{p}$. Note this algorithm uses again as an auxiliary element a quadratic nonresidue.
Iteratively applying it we can use ${\alpha_{(j)}}^{n/2^{i}} \pmod{p_{j}}$ the $2^{i}$\=/th root of $a$ to compute its square root ${\alpha_{(j)}}^{n/2^{i+1}}$, until we finally reach $\alpha_{(j)}$, from which we can recover $\alpha$.

Coming back again to the case $x^{n}+1$, we describe every step  to compute $(\alpha,\omega)$ in \Cref{alg:alphaOmega}.

\begin{algorithm}[H]
    \linespread{1.35}\selectfont
     \caption{\textsc{Computation of $\alpha$ and $\omega$}}\label{alg:alphaOmega}
  \SetAlgoLined{}
  \KwInput{A power of two $n$ and a modulus $m$ (with known factorization)}
  \KwResult{Suitable $(\alpha,\omega)$ roots of $-1$ and $1$}
  Let $m={p_{1}}^{e_1} \dots {p_{k}}^{e_k}$ be the prime decomposition of $m$.\\
  \nl{} Ensure $2n \vert p_{j} -1$ for every prime and abort otherwise.\\
  \For{$j \gets 1$ \KwTo{} $k$}{
    \nl{} \tcc{obtain a quadratic nonresidue in $\Z_{p_{j}}$}
    $tests = \mathsf{False}$\tcp*{whether a candidate is a nonresidue}
    \While{\KwNot{} $tests$}{
      $u_{j} \dleftarrow \Z_{p_{j}}$\tcp*{choose a nonresidue candidate}
      \If{$u_{j}^{(p_{j}-1)/2} \equiv -1 \pmod{p_{j}}$}{
        $tests = \mathsf{True}$\tcp*{it is a nonresidue}
        }
    }
    \nl{} \tcc{compute $\alpha$ and $\omega \pmod{{p_{j}}^{e_{j}}}$}
    $\alpha_{(j)} = u_{j}^{(p-1)/2n}$\tcp*{compute $\alpha \pmod{p_{j}}$}
    $\omega_{(j)} = u_{j}^{(p-1)/n}$\tcp*{compute $\omega \pmod{p_{j}}$}
    $c,aux = \textsc{ExtEuclides}(n\alpha_{(j)}^{2n-1},p_{j})$\tcp*{compute $\overline{f'(\alpha_{(j)})}$}
    $d,aux = \textsc{ExtEuclides}(n\omega_{(j)}^{n-1},p_{j})$\tcp*{compute $\overline{f'(\omega_{(j)})}$}
    \For(\tcp*[f]{apply Hensel Lemma}){$e \gets 2$ \KwTo{} $e_{j}$}{
      $\alpha_{(j)} \gets \alpha_{(j)} - (\alpha_{(j)}^{n}+1)c \rem{{p_{j}}^{e}}$\tcp*{lift from ${{p_{j}}^{e-1}}$ to ${{p_{j}}^{e}}$}
      $\omega_{(j)} \gets \omega_{(j)} - (\omega_{(j)}^{n}-1)d \rem{{p_{j}}^{e}}$\tcp*{lift from ${{p_{j}}^{e-1}}$ to ${{p_{j}}^{e}}$}
    }
  }
  \nl{} Reconstruct $\alpha$ from $(\alpha_{(1)},\dots,\alpha_{(k)})$\tcp*{via the CRT}
  \nl{} Reconstruct $\omega$ from $(\omega_{(1)},\dots,\omega_{(k)})$\tcp*{via the CRT}
  \Return{$(\alpha,\omega)$}
\end{algorithm}

%% file: sections/generalizations.tex

\section{FFT Generalizations}\label{sec:generalizations}

As we have seen in \Cref{sec:suitableSets} we only have suitable evaluation points in the ring $\Z_{m}[x]/\left<x^{n}+1\right>$ if $m={p_{1}}^{e_1} \dots {p_{k}}^{e_k}$ is such that every $p_{i} \equiv 1 \pmod{2n}$.

These congruences are deeply related to the factorization of $x^n+1$ (irreducible in $\Z[x]$) when considered modulo $m$.
It has been described in~\cite{EC:LyuSei18} where~\Cref{thm:partiallySplitting} is presented.

\begin{theorem}[Corollary~1.2 in~\cite{EC:LyuSei18}]\label{thm:partiallySplitting}
  Let $n \geq d > 1$ be powers of $2$ and $p \equiv 2d + 1 \pmod{4d}$ be a prime. Then the polynomial $x^n + 1$ factors as
  \[x^n+1 \equiv \prod_{j=0}^{d-1}\left(x^{n/d}-\alpha_{j}\right) \pmod{p}\]
  for distinct $\alpha_{j} \in \Z_{p}^{*}$, where $x^{n/d}-\alpha_{j}$ are irreducible in $\Z_{p}[x]$.
\end{theorem}

\Cref{thm:partiallySplitting} can also be generalized to a not necessarily prime modulus $m$ taking advantage of the results discussed in the previous sections.

\begin{theorem}[Generalization of \Cref{thm:partiallySplitting} to a not necessarily prime modulus $m$]\label{thm:partiallySplittingNonprime}
  Let $n \geq d > 1$ be powers of $2$, $m={p_{1}}^{e_1} \dots {p_{k}}^{e_k}$ the prime decomposition of $m$ such that $p_{i} \equiv 2d + 1 \pmod{4d}$. Then the polynomial $x^n + 1$ factors as
  \[x^n+1 \equiv \prod_{j=0}^{d-1}\left(x^{n/d}-\alpha_{j}\right) \pmod{m}\]
  for distinct $\alpha_{j} \in \Z_{m}^{*}$.
\end{theorem}

\begin{proof}
  Notice first $p_{i} \equiv 2d + 1 \pmod{4d}$ implies $p_{i} \equiv 1 \pmod{2d}$, therefore by the results discussed in \Cref{sec:existenceAndConstruction} we know there exists a twofold set of $d$\=/th roots of $-1$ with invertible differences $\alpha_{0},\dots,\alpha_{d-1}$.
  From \Cref{lemma:invertibilityEvaluationPoints} we know each $\alpha_{i} \in \Z_{m}^{\ast}$. We also know from \Cref{lemma:opuesto} that

  \[ \prod_{j=0}^{d-1}\left(x^{n/d}-\alpha_{j}\right) \equiv \prod_{j=0}^{d/2-1}\left(x^{n/d}-\alpha_{j}\right)\left(x^{n/d}-\alpha_{\overline{\jmath}}\right) \equiv \prod_{j=0}^{d/2-1}\left(x^{2n/d}-{\alpha_{j}}^{2}\right). \]

  Given that squares preserve the initial properties (as seen in \Cref{{prop:squaring}}) iteratively applying the same idea we finally get $\displaystyle\prod_{j=0}^{d-1}\left(x^{n/d}-\alpha_{j}\right) \equiv x^{n} + 1 \pmod{m}$ as desired.
\end{proof}

This implies that, under the necessary conditions for an FFT multiplication algorithm, $x^{n}+1$ fully splits in linear factors when considered modulo $m$. However sometimes this is not the case and we are enforced to use modulus that specifically require $x^{n}+1$ to split in a smaller number of factors.

That is the case of some cryptographic constructions that use \Cref{thm:partiallySplitting} from~\cite{EC:LyuSei18} (or similar versions) to guarantee the invertibility of particular subsets of elements in $\Z_{m}[x]/\left<x^{n}+1\right>$.

The condition $p_{i} \equiv 2d + 1 \pmod{4d}$ implies that there are $d$\=/th roots of $-1$ in $\Z_{m}$, but no $2d$\=/th roots of $-1$. Therefore if $d<n$ no suitable evaluation points exist.

However having only $d$\=/th roots of $-1$ does not prevent us from finding a reasonably efficient multiplication algorithm. We cannot complete the recursion strategy but we still can partially use it.

To do so we just need to apply a technique called $n/d$\=/degree \emph{striding} (as described in~\cite{Bernstein01}), mapping our polynomials to a more convenient ring, defining an auxiliary new variable $y=x^{n/d}$. We can always consider $R[x]/\left<x^{n}+1\right>$ as a subring of $R[x,y]/\left<x^{n/d}-y,x^{n}+1\right>$, and observe we can also describe this second ring as $R[x][y]/\left<x^{n/d}-y,y^{d}+1\right>$.

With this simple change of variables we can now represent our original polynomial as a polynomial in $y$ with $y$\=/degree bounded by $d$ that has as coefficients polynomials in $R[x]$ of $x$\=/degree bounded by $n/d$. As a polynomial in $y$ it satisfies all required conditions as we are only considering modulus $y^{d}+1$ and the new ring $R'=R[x]$ does contain $d$ evaluation points $\alpha_{0},\dots,\alpha_{d-1}$ satisfying \Cref{cond:evalAtRoots,cond:invertibleDifferences,cond:twofold}.

With these $d$\=/th roots of $-1$ we can efficiently use $g(x,y)$ and $h(x,y)$ to compute evaluations $\{g(x,\alpha_{i})\}$ and $\{h(x,\alpha_{i})\}$, do a pointwise product in $R[x]$ (using an auxiliary multiplication algorithm, such as (Dual) Karatsuba) and invert the transform to recover the product $(g \cdot h)(x,y)$ that directly gives us the desired solution substituting again $y$ with $x^{n/d}$.

Observe that, the same way we have to choose an $\alpha$ that is a $d$\=/th root of $-1$ to make the evaluation compatible with the quotient $\left<y^{d}+1\right>$ the other quotient $\left<x^{n/d}-y\right>$ implies that when computing the evaluation of the variable $y$ at $\alpha$ we obtain as a result a polynomial in $R[x]/\left<x^{n/d}-\alpha\right>$ and therefore is only defined modulo $x^{n/d}-\alpha$. This is an artifact of the technique due to the additional variable introduced that has no impact in the process but helps us keep these $x$-polynomials bounded when computing their products.

The running time for both the transform and the anti-transform is now $\bigO{n\log{d}}$ and the $d$ products of polynomials of degree $n/d$ that require each $\bigO{{(n/d)}^{\log{3}}}$ for a total of $\bigO{n\log{d} + d{(n/d)}^{\log{3}}}$.

From an abstract point of view we could directly apply \Cref{alg:efficientMultiplication}, as it was described for a general ring $R$ and therefore we could use it just taking into account to which ring each element belongs. However, for the sake of readability, we explicitate this generalization in \Cref{alg:generalizedEfficientMultiplication}.

\begin{algorithm}[H]
  \linespread{1.35}\selectfont
   \caption{\textsc{Generalized Efficient FFT Multiplication}}\label{alg:generalizedEfficientMultiplication}
\SetAlgoLined{}
\KwInput{Two polynomials $g(x)$, $h(x)$ of degree bounded by $n$}
\KwAuxiliary{A twofold set $\alpha_{0},\dots,\alpha_{d-1}$ of $d$\=/th roots of $-1$ with invertible differences}
\KwResult{The product $(g \cdot h)(x)$ of $g(x)$ and $h(x)$ in $\Z_{m}[x]/\left< x^n + 1 \right>$}
\nl{} $\widehat{g}(y) \coloneqq g(x) \rem{x^{n/d}-y}$ \tcp*{polynomial in $\left(\Z_{m}[x]\right)[y]$}
\nl{} $\widehat{h}(y) \coloneqq h(x) \rem{x^{n/d}-y}$ \tcp*{polynomial in $\left(\Z_{m}[x]\right)[y]$}
\nl{} $\vec{g} \coloneqq \textsc{FFT}(\widehat{g}(y),\alpha_{0}, \dots, \alpha_{d-1})$ \tcp*{vector of polynomials in $\Z_{m}[x]$}
\nl{} $\vec{h} \coloneqq \textsc{FFT}(\widehat{h}(y),\alpha_{0}, \dots, \alpha_{d-1})$ \tcp*{vector of polynomials in $\Z_{m}[x]$}
Define $\vec{f}$ a vector of polynomials in $\Z_{m}[x]$ of size $d$.\\
\nl{} \For{$i \gets 0$ \KwTo{} $d-1$}{
  $\vec{f}[i] \coloneqq \textsc{Karatsuba}(\vec{g}[i],\vec{h}[i]) \rem x^{n/d} - \alpha_{i}$  \\
}
\nl{} $\widehat{f}(y) \coloneqq \textsc{IFFT}(\vec{f},\alpha_{0}, \dots, \alpha_{d-1})$ \\
\nl{} $f(x) \coloneqq \widehat{f}(y) \rem y - x^{n/d}$ \\
\Return{$f(x)$}
\end{algorithm}

\subsection{Fast Chinese Remaindering}

As we mentioned in the introduction, this particular issue of partially splitting rings, where we cannot directly apply the original full FFT to the initial polynomials, has been studied in~\cite{EC:LyuSei18} for rings with prime modulus from a different point of view, considering FFT-like algorithms for efficiently applying the Chinese Reminder Theorem~\cite{EC:LyuPeiReg13,gathen_gerhard_2013}.

The main idea of these CRT approaches is to consider evaluations at $\alpha_{i}$ as representatives for $g(x) \pmod{x-\alpha_{i}}$, sufficient for determining $g(x)$ via the CRT since $x^{n}+1 \equiv \prod_{j=0}^{n-1}\left(x-\alpha_{j}\right) \pmod{m}$ (as we know from \Cref{,thm:partiallySplitting,thm:partiallySplittingNonprime}).
The same kind of recursions apply, as both $g(x) \pmod{x-\alpha_{i}}$ and $g(x) \pmod{x-\alpha_{\overline{\imath}}}$ can be computed from $g(x) \pmod{x^{2}-{\alpha_{i}}^{2}}$ (in an equivalent manner to what we saw in \Cref{sec:efficientTrans}).

Given an twofold set $\alpha_{0},\dots,\alpha_{d-1}$ of $d$\=/th roots of $-1$ what we do in \Cref{alg:generalizedEfficientMultiplication} is precisely computing the reminders $\prem{x^{n/d} - \alpha_{i}}$ that determine the original polynomials via the CRT\@.

However we believe our presentation is still more direct and informative, since the generalization to a not necessarily prime modulus in a partially splitting ring, not explored in~\cite{EC:LyuSei18}, comes completely for free, while an interpretation using the Chinese Reminder Theorem would technically require a much more involved analysis when the ring is not a Principal Ideal Domain, or not even a Unique Factorization Domain, as in order to verify the hypothesis of the theorem one should check whether some ideals are comaximal in order to ensure that every mapping is indeed an isomorphism, and some conditions (as the requirement for $m$ to be odd) would only appear as artifacts of the construction.

  Recall the evaluation $g(\alpha)$ of any polynomial in $\alpha$ is equivalent to computing its remainder after dividing by $x-\alpha$ (this is known as the Polynomial Remainder Theorem).

  All the evaluations at each of the $\alpha_{i}$ are just $g(\alpha_{i}) = g(x) \prem{x-\alpha_{i}}$. We can see the vector whose components are these evaluations as the CRT representation of the polynomial. Pointwise multiplication of these vectors of evaluations for two polynomials $g(x)$ and $h(x)$ is just a multiplication in the CRT domain and the interpolation consists of recovering the polynomial coefficients from the CRT representations.

  It seems just like a different interpretation of the same idea. It allows the same analysis as the CRT representation over a set of factors can be computed even if $x^n-a$ does not fully split.

  %
  %
  %

  The important point is that this argument works even if $d \neq n$, and we could use it to represent any polynomial with $d$ remainders $\prem{x^{n/d} - \alpha_{i}}$, compute the pointwise products among polynomials of degree bounded by $n/d$ and then recover back the product polynomial modulo ${x^n-a}$.

  \begin{theorem}[CRT for {$\Z_{m}[x]/\left<x^{n}-a\right>$}]\label{thm:splittingCRT}
    Let $x^{n}-a = \prod_{i=0}^{d-1}(x^{n/d}-\alpha_{i})$ where ${\{\alpha_{i}\}}_{i}$ are a twofold set of $d$\=/th roots of $a$ with differences invertible in $\Z_{m}$,
    then $\Z_{m}[x]/\left<x^n-a\right> \cong \Z_{m}[x]/\left<x^{n/d}-\alpha_{0}\right> \times \cdots \times \Z_{m}[x]/\left<x^{n/d}-\alpha_{d-1}\right>$.
  \end{theorem}

  \begin{proof}
    For convenience we are going to label the $d$ roots as $\alpha_{b_{\log(d)-1}\dots b_{0}}$ as we did in $\Cref{fig:binarytree}$.

    From an argument analogous to \Cref{thm:partiallySplittingNonprime} we know that we can write $x^{n/{2^k}} - \alpha_{b_{k-1}\dots b_{0}} = (x^{n/{2^{k+1}}} - \alpha_{0b_{k-1}\dots b_{0}})(x^{n/{2^{k+1}}} - \alpha_{1b_{k-1}\dots b_{0}})$.

    We have to prove that $\Z_{m}[x] \Big/ \left<x^{n/2^k} - \alpha_{b_{k-1}\dots b_{0}}\right>$ is isomorphic to
    \[\Z_{m}[x] \Big/ \left<x^{n/2^{k+1}} - \alpha_{0b_{k-1}\dots b_{0}}\right>
    \times \Z_{m}[x]\Big/\left<x^{n/2^{k+1}} - \alpha_{1b_{k-1}\dots b_{0}}\right>.
    \]

    We can define a map from $\Z_{m}[x]$ to
    \[\Z_{m}[x] \Big/ \left<x^{n/2^{k+1}} - \alpha_{0b_{k-1}\dots b_{0}}\right>
    \times \Z_{m}[x]\Big/\left<x^{n/2^{k+1}} - \alpha_{1b_{k-1}\dots b_{0}}\right>\]
     by computing $\trem$, and the kernel would be
     \[\left<x^{n/2^{k+1}} - \alpha_{0b_{k-1}\dots b_{0}}\right> \bigcap \left<x^{n/2^{k+1}} - \alpha_{1b_{k-1}\dots b_{0}}\right>.\]

     That is, the map is also well defined from
     \[\Z_{m}[x] \Big/ \left<x^{n/2^{k+1}} - \alpha_{0b_{k-1}\dots b_{0}}\right> \bigcap \left<x^{n/2^{k+1}} - \alpha_{1b_{k-1}\dots b_{0}}\right>.\]

     So far this discussion has been completely general. However for this particular polynomials we have
     \[\left<x^{n/2^{k+1}} - \alpha_{0b_{k-1}\dots b_{0}}\right> + \left<x^{n/2^{k+1}} - \alpha_{1b_{k-1}\dots b_{0}}\right> \cong \Z_{m}[x].\]
     This is the case since
     \[(x^{n/2^{k+1}} - \alpha_{0b_{k-1}\dots b_{0}}) - (x^{n/2^{k+1}} - \alpha_{1b_{k-1}\dots b_{0}}) = \alpha_{1b_{k-1}\dots b_{0}} - \alpha_{0b_{k-1}\dots b_{0}}\]
     which is invertible in $\Z_{m}$ (as every difference of roots is invertible), implying
     \[1 \in \left<x^{n/2^{k+1}} - \alpha_{0b_{k-1}\dots b_{0}}\right> + \left<x^{n/2^{k+1}} - \alpha_{1b_{k-1}\dots b_{0}}\right>.\]

     We can explicitly write this saying there are two polynomials $g(x)$ and $h(x)$ such that
     \[(x^{n/2^{k+1}} - \alpha_{0b_{k-1}\dots b_{0}})g(x) + (x^{n/2^{k+1}} - \alpha_{1b_{k-1}\dots b_{0}})h(x)=1.\]

     On the one hand this implies that the map is surjective. From the previous identity we know any pair of polynomials $(a(x),b(x))$ has a preimage $b(x)g(x)(x^{n/2^{k+1}} - \alpha_{0b_{k-1}\dots b_{0}})+a(x)h(x)(x^{n/2^{k+1}} - \alpha_{1b_{k-1}\dots b_{0}})$.

     On the other hand this implies
     \[\left<x^{n/2^{k+1}} - \alpha_{0b_{k-1}\dots b_{0}}\right> \bigcap \left<x^{n/2^{k+1}} - \alpha_{1b_{k-1}\dots b_{0}}\right> = \left<x^{n/2^k} - \alpha_{b_{k-1}\dots b_{0}}\right>.\]

     The right-hand side is directly a subset of the left-hand side. To see the other inclusion we can check that for any
     \[a(x) \in \left<x^{n/2^{k+1}} - \alpha_{0b_{k-1}\dots b_{0}}\right> \bigcap \left<x^{n/2^{k+1}} - \alpha_{1b_{k-1}\dots b_{0}}\right>,\]
     that is, there are polynomials $b(x)$ and $c(x)$ such that
     \[a(x)=b(x)(x^{n/2^{k+1}} - \alpha_{0b_{k-1}\dots b_{0}}) = c(x)(x^{n/2^{k+1}} - \alpha_{1b_{k-1}\dots b_{0}}),\]
     it is also true that
     \begin{align*}
       a(x) &= a(x) \cdot 1 \\
        &= a(x)(x^{n/2^{k+1}} - \alpha_{0b_{k-1}\dots b_{0}})g(x) + a(x)(x^{n/2^{k+1}} - \alpha_{1b_{k-1}\dots b_{0}})h(x)\\
        &=c(x)(x^{n/2^{k+1}} - \alpha_{1b_{k-1}\dots b_{0}})(x^{n/2^{k+1}} - \alpha_{0b_{k-1}\dots b_{0}})g(x) \\
        &\phantomrel{=}+
        b(x)(x^{n/2^{k+1}} - \alpha_{0b_{k-1}\dots b_{0}})(x^{n/2^{k+1}} - \alpha_{1b_{k-1}\dots b_{0}})h(x)\\
        &= (c(x)g(x)+b(x)h(x))(x^{n/2^k} - \alpha_{b_{k-1}\dots b_{0}})
     \end{align*}
     and therefore $a(x) \in \left<x^{n/2^k} - \alpha_{b_{k-1}\dots b_{0}}\right>$.

     Summing up the desired mapping is an isomorphism.

     The main issue here is that CRT is usually defined for principal ideal domains or at least unique factorization domains. If it is not the case, such as with our construction, we have to specifically check these additional properties, such as $\left<x^{n/2^{k+1}} - \alpha_{0b_{k-1}\dots b_{0}}\right>$ and $\left<x^{n/2^{k+1}} - \alpha_{1b_{k-1}\dots b_{0}}\right>$ being comaximal.
  \end{proof}

    The task we have is defined as follows. We have a polynomial $g(x) \in \Z_{m}[x]/\left<x^n-a\right>$ and a twofold set $\{\alpha_{b_{\log(d)-1}b_{\log(d)-2}\dots b_{0}}\}$ of $d$\=/th roots of $a$ with invertible differences.

    Our goal is to compute all $g_{b_{\log(d)-1}\dots b_{0}}(x) = g(x) \prem{x^{n/d}-\alpha_{b_{\log(d)-1}\dots b_{0}}}$.

    To do so we start computing $g_{0}(x) = g(x) \prem{x^{n/2}-\alpha_{0}}$ and $g_{1}(x) = g(x) \prem{x^{n/2}-\alpha_{1}}$.
    This requires $\bigO{n}$ operations.

    Then we notice $g_{0b_{0}}(x) = g(x) \prem{x^{n/4}-\alpha_{0b_{0}}} = g_{b_{0}}(x) \prem{x^{n/4}-\alpha_{0b_{0}}}$ and $g_{1b_{0}}(x) = g(x) \prem{x^{n/4}-\alpha_{1b_{0}}} = g_{b_{0}}(x) \prem{x^{n/4}-\alpha_{1b_{0}}}$.

    In general we can recursively compute $g_{b_{i}b_{i-1} \dots b_{0}}(x) = g(x) \prem{x^{n/2^{i+1}}-\alpha_{b_{i}b_{i-1} \dots b_{0}}} = g_{b_{i-1} \dots b_{0}}(x) \prem{x^{n/2^{i+1}}-\alpha_{b_{i}b_{i-1} \dots b_{0}}}$. Therefore at this $i$\=/th level computing each remainder takes $\bigO{n/2^{i}}$ operations and has to be done $\bigO{2^{i}}$ times, for a total cost of $\bigO{n}$.

    Observe computing $g_{b_{i-1} \dots b_{0}}(x) \prem{x^{n/2^{i+1}}-\alpha_{0b_{i-1} \dots b_{0}}}$ is done taking the lower coefficients of $g_{b_{i-1} \dots b_{0}}(x)$ and adding the higher coefficients multiplied by $\alpha_{0b_{i-1} \dots b_{0}}$.

    The same way, since $\alpha_{1b_{i-1} \dots b_{0}} = - \alpha_{0b_{i-1} \dots b_{0}}$  we have that $g_{b_{i-1} \dots b_{0}}(x) \prem{x^{n/2^{i+1}}-\alpha_{1b_{i-1} \dots b_{0}}}$ is computed taking the lower coefficients of $g_{b_{i-1} \dots b_{0}}(x)$ and subtracting the higher coefficients multiplied by $\alpha_{0b_{i-1} \dots b_{0}}$.

    Notice the multiplications are the same (and could be reused) and the only difference is that we add or subtract depending on the case.

    The number of levels is $\log(d)$ and we end up with $\bigO{n\log(d)}$ operations to compute the CRT representation of $g(x)$ taking modulus over the $d$ different polynomials of degree $n/d$.

    \begin{algorithm}[H]
      \linespread{1.35}\selectfont
       \caption{\textsc{FFT} (CRT)}\label{alg:fftCRTalt}
    \SetAlgoLined{}
    \KwInput{A polynomial $g(x)$ of degree bounded by $n$ and a twofold set $\alpha_{0}, \dots, \alpha_{d-1}$ of $d$\=/th roots of $a$ with invertible differences}
    \KwResult{Remainders of $g(x)$ when divided by $x^{n/d} -\alpha_{0}, \dots,x^{n/d} - \alpha_{d-1}$}
    \lIf{$r = 1$}{\Return{$g$}}r
    $\{g_{b_{\log(d)-2} \dots b_{0}}\} \coloneqq \textsc{FFT}(g,{\alpha_{0}}^{2}, \dots, {\alpha_{d/2-1}}^{2})$\\
    \For{$b_{\log(d)-2} \dots b_{0} \gets {\{0,1\}}^{\log(d)-1}$}{
      Split $g_{b_{\log(d)-2} \dots b_{0}}$ into $g_{b_{\log(d)-2} \dots b_{0}}^{L}$ and $g_{b_{\log(d)-2} \dots b_{0}}^{H}$\\
      $g_{0b_{\log(d)-2} \dots b_{0}} \coloneqq g_{b_{\log(d)-2} \dots b_{0}}^{L} + g_{b_{\log(d)-2} \dots b_{0}}^{H} \cdot \alpha_{0b_{\log(d)-2} \dots b_{0}}$\\
      $g_{1b_{\log(d)-2} \dots b_{0}} \coloneqq g_{b_{\log(d)-2} \dots b_{0}}^{L} - g_{b_{\log(d)-2} \dots b_{0}}^{H} \cdot \alpha_{0b_{\log(d)-2} \dots b_{0}}$\\
    }
    \Return{$\{g_{b_{\log(d)-1}b_{\log(d)-2} \dots b_{0}}\}$}
    \end{algorithm}

    At this point we could use Karatsuba's algorithm to multiply them, that is, a cost of $\bigO{d{(n/d)}^{\log_{2}(3)}}$.

    Then we have to invert this operations. To do so we could follow the same ideas inverting the operations at each level. To recover the lower part of $g_{b_{i-1} \dots b_{0}}(x)$ we add $g_{0b_{i-1} \dots b_{0}}(x) + g_{1b_{i-1} \dots b_{0}}(x)$ and divide by two (that is, multiply each coefficient by $2^{-1}$). To recover the upper part we now subtract them computing $g_{0b_{i-1} \dots b_{0}}(x) - g_{1b_{i-1} \dots b_{0}}(x)$ and multiply each coefficient by $2^{-1}{(\alpha_{0b_{i-1} \dots b_{0}})}^{-1}$. The cost of this operations is again $\bigO{n\log(d)}$.

    As we have to divide by $2$ at each level one could just skip this step and divide by $d$ at the end.

    \begin{algorithm}[H]
      \linespread{1.35}\selectfont
       \caption{\textsc{IFFT} (CRT)}\label{alg:ifftCRTalt}
    \SetAlgoLined{}
    \KwInput{A CRT representation $g_{0} \dots g_{d-1}$ of a polynomial $g(x)$ of degree bounded by $n$, and a twofold set $\alpha_{0}, \dots, \alpha_{d-1}$ of $d$\=/th roots of $a$ with invertible differences}
    \KwResult{The polynomial $g(x)$}
    \lIf{$r=2$}{\Return{$2^{-1}\left((g_{0}+g_{1})+x^{n/2}(g_{0}+g_{1})\cdot {\alpha_{0}}^{-1}\right)$}}
    \For{$b_{\log(d)-2} \dots b_{0} \gets {\{0,1\}}^{\log(d)-1}$}{
       $\begin{aligned}
         &g_{b_{\log(d)-2} \dots b_{0}} \coloneqq 2^{-1}(g_{0b_{\log(d)-2} \dots b_{0}}+g_{1b_{\log(d)-2} \dots b_{0}}) + \\
          &\phantomrel{\coloneqq} x^{n/2} 2^{-1} (g_{0b_{\log(d) - 2} \dots b_{0}}-g_{1b_{\log(d)-2} \dots b_{0}})\cdot {\alpha_{0b_{\log(d)-2} \dots b_{0}}}^{-1}
       \end{aligned}$
    }
    \Return{${\normalfont\textsc{IFFT}}(\{g_{b_{\log(d)-2} \dots b_{0}}\},{\alpha_{0}}^{2}, \dots, {\alpha_{d/2-1}}^{2})$}
    \end{algorithm}

  This alternative interpretation has then same properties and can also be used to design the same efficient multiplication algorithms. We however believe the previous presentation was more insightful.

%% file: sections/resultsAndDiscussion.tex

\section{Results and Discussion}

The main result can be summarized in the following way, in order to design an efficient multiplication algorithm in $\Z_{m}[x]/\left<x^{n}-a\right>$ via an FFT the necessary and sufficient condition is to have a set of $n$ different $n$\=/th roots of $a$ (\Cref{cond:evalAtRoots}, so that multiplication is compatible with congruence classes) with invertible differences (\Cref{cond:invertibleDifferences}, so that the inverse transform is defined) such that its recursive squares are equal two by two (\Cref{cond:twofold}, so that the computation can be efficiently done recursively).

This characterization is similar but not equivalent to the usual characterization with roots of unity (\Cref{cond:alphaOmega}), which is sufficient but not necessary. We have proven however that, as \Cref{thm:relationConditionsAlphaOmega} states, restricting to sets satisfying \Cref{cond:alphaOmega} does not decrease the applicability of these efficient multiplication algorithms as these $(\alpha,\omega)$\=/sets exist if and only if the necessary and sufficient sets satisfying \Cref{cond:evalAtRoots,cond:invertibleDifferences,cond:twofold} exist.

As intermediate result we have also proven that these properties are indeed independent in the general case, and we do believe that this analysis might help to clarify whether some considerations and conditions usually stated in the folklore are fundamental considerations about the algebraic structure or just conventions for a particular instantiation on a particular setting (as it is the case with roots of unity).

This framework is also a general introduction to FFT\=/multiplication from a rigorous mathematical point of view while still keeping it readable for an audience not familiarized with more advanced algebraic considerations.

For example our analysis directly generalizes, as we have seen in~\Cref{sec:generalizations}, to a ring where $x^{n}-a$ does not fully split and $\Z_{m}$ is not a field, while alternative interpretations are much more delicate to work with.